%% file: TSP_arxive_4.tex
\documentclass[journal,table]{IEEEtran}

\usepackage{fixltx2e}
\usepackage{amsmath, amsthm, amssymb, bm}
\usepackage{epsfig}
\usepackage{amsfonts}
\usepackage{color}
\usepackage{enumerate}
\usepackage{relsize}
\usepackage[T4]{fontenc}
\usepackage[T1]{fontenc}
\usepackage{amsmath}
\usepackage{tikz,pgfplots}
\usepackage{caption}
\usepackage{stackengine}
\usepackage{rotating}
\usepackage{lipsum}
\usepackage{algorithmic}
\usepackage{algorithm}
\usepackage{afterpage}
\usepackage{cite}
\usepackage{epstopdf}
\usepackage{footnote}
\usepackage[bottom]{footmisc}
\usepackage{multicol}
\newtheorem{theorem}{Theorem}%[section]
\newtheorem{property}{Property}

\newtheorem{remark}{Remark}

\usepackage{subfigure}
\usepackage{chngcntr}
\counterwithout{figure}{section}
\counterwithout{figure}{subsection}
\usepackage{ mathrsfs }
\usepackage{dblfloatfix}
\newlength\figureheight
\newlength\figurewidth

\nointerlineskip

\include{macros}

\begin{document}
%\onecolumn
\title{Perturbation-Based Regularization for Signal Estimation in Linear Discrete Ill-posed Problems}

%Copyright (c) 2014 IEEE. Personal use of this material is permitted. However, permission to use this material for any other purposes must be obtained from the % IEEE by sending a request to pubs-permissions@ieee.org.
\author{ Mohamed~Suliman,~\IEEEmembership{Student~Member, IEEE}, Tarig~Ballal,~\IEEEmembership{Member, IEEE}, and Tareq~Y.~Al-Naffouri,~\IEEEmembership{Member, IEEE}% <-this % stops a space

%\vspace{-7pt}

\thanks{M. Suliman, T. Ballal, and T. Y. Al-Naffouri are with the Computer, Electrical, and Mathematical Sciences and Engineering (CEMSE) Division, King Abdullah University of Science and Technology (KAUST), Thuwal, Makkah Province, Saudi Arabia. e-mails:$\{$mohamed.suliman, tarig.ahmed, tareq.alnaffouri$\}$@kaust.edu.sa.}}

\maketitle

\begin{abstract}
Estimating the values of unknown parameters from corrupted measured data faces a lot of challenges in ill-posed problems. In such problems, many fundamental estimation methods fail to provide a meaningful stabilized solution. In this work, we propose a new regularization approach and a new regularization parameter selection approach for linear least-squares discrete ill-posed problems. The proposed approach is based on enhancing the singular-value structure of the ill-posed model matrix to acquire a better solution. Unlike many other regularization algorithms that seek to minimize the estimated data error, the proposed approach is developed to minimize the mean-squared error of the estimator which is the objective in many typical estimation scenarios. The performance of the proposed approach is demonstrated by applying it to a large set of real-world discrete ill-posed problems. Simulation results demonstrate that the proposed approach outperforms a set of benchmark regularization methods in most cases. In addition, the approach also enjoys the lowest runtime and offers the highest level of robustness amongst all the tested benchmark regularization methods.
\end{abstract}
\section{Introduction}
\label{sec:intro}
We consider the standard problem of recovering an unknown signal $\xv_{0} \in \mathbb{R}^{n}$ from a vector $\yv \in \mathbb{R}^{m}$ of $m$ noisy, linear observations given by $\yv = \Am \xv_{0} + \zv$. Here, $\Am \in \mathbb{R}^{m\times n}$ is a known linear measurement matrix, and, $\zv \in \mathbb{R}^{m\times 1}$ is the noise vector; the latter is assumed to be additive white Gaussian noise (AWGN) vector with unknown variance $\sigma_{\zv}^{2}$ that is independent of $\xv_{0}$. Such problem has been extensively studied over the years due to its obvious practical importance as well as its theoretical interest \cite{kailath2000linear, poor2013introduction, groetsch1993inverse}. It arises in many fields of science and engineering, e.g., communication, signal processing, computer vision, control theory, and economics.

Over the past years, several mathematical tools have been developed for estimating the unknown vector $\xv_{0}$. The most prominent approach is the ordinary least-squares (OLS)\cite{kay2013fundamentals} that finds an estimate $\hat{\xv}_{\text{OLS}}$ of $\xv_{0}$ by minimizing the Euclidean norm of the residual error
\begin{equation}
\label{eq:ls problem}
\hat{\xv}_{\text{OLS}} = \argmin_{\xv}  || \yv - \Am \xv||^{2}_{2}.
\end{equation}
The behavior of the OLS has been extensively studied in the literature and it is now very well understood. In particular, if $\Am$ is a full column rank, (\ref{eq:ls problem}) has a unique solution given by  
\begin{eqnarray}
\label{eq:pure LS solution}
{\hat{\xv}}_{\text{OLS}}
&=& \left(\Am^{T}  \Am\right)^{-1}  \Am^{T} \yv =\Vm \Sigmam^{-1} \Um^{T} \yv,  %{\dagger}
\end{eqnarray}
where $\Am = \Um \Sigmam \Vm^{T} = \sum_{i=1}^{n} \sigma_{i} \uv_{i} \vv_{i}^{T}$ is the singular value decomposition (SVD) of $\Am$, $\uv_{i}$ and $\vv_{i}$ are the left and the right orthogonal singular vectors, while the singular values $\sigma_{i}$ are assumed to satisfy $\sigma_{1} \ge \sigma_{1} \geq \dotsi \geq \sigma_{n}$. A major difficulty associated with the OLS approach is in discrete ill-posed problems. A problem is considered to be well-posed if its solution always exists, unique, and depends continuously on the initial data. Ill-posed problems fail to satisfy at least one of these conditions \cite{fischler2014readings}. In such problems, the matrix $\Am$ is ill-conditioned and the computed LS solution in (\ref{eq:pure LS solution}) is potentially very sensitive to perturbations in the data such as the noise $\zv$ \cite{kilmer2001choosing}.

Discrete ill-posed problems are of great practical interest in the field of signal processing and computer vision \cite{piotrowski2008mv, liu2008kernel, bertero1988ill, poggio1985computational}. They arise in a variety of applications such as computerized tomography \cite{natterer1986mathematics}, astronomy \cite{craig1986inverse}, image restoration and deblurring \cite{katsaggelos1991regularized, hansen2006deblurring}, edge detection \cite{torre1986edge}, seismography \cite{scales1988robust}, stereo matching \cite{scharstein2002taxonomy}, and the computation of lightness and surface reconstruction \cite{blanz2004statistical}. Interestingly, in all these applications and even more, data are gathered by convolution of a noisy signal with a detector \cite{aster2005parameter, hansen1993use}. A linear representation of such process is normally given by
\begin{equation}
\label{eq:kernal equation}
\int_{b_{1}}^{b_{2}} a\left(s,t\right) \xv_{0}\left(t\right) \text{dt} = \yv_{0}\left(s\right) + \zv\left(s\right) = \yv\left(s\right),
\end{equation} 
where $ \yv_{0}\left(s\right)$ is the true signal, while the kernal function $a\left(s,t\right)$ represents the response. It is shown in \cite{chen1993new} how a problem with a formulation similar to (\ref{eq:kernal equation}) fails to satisfy the well-posed conditions introduced above. The discretized version of (\ref{eq:kernal equation}) can be represented by $\yv = \Am \xv_{0} + \zv$.

To solve ill-posed problems, regularization methods are commonly used. These methods are based on introducing an additional prior information in the problem. All regularization methods are used to generate a reasonable solution for the ill-posed problem by replacing the problem with a well-posed one whose solution is acceptable. This must be done after careful analysis to the ill-posed problem in terms of its physical plausibility and its mathematical properties.

Several regularization approaches have been proposed throughout the years. Among them are the truncated SVD \cite{varah1983pitfalls}, the maximum entropy principle \cite{smith2013maximum}, the hybrid methods \cite{hanke1993regularization}, the covariance shaping LS estimator \cite{eldar2003covariance}, and the weighted LS \cite{eldar2007improvement}. The most common and widely used approach is the regularized M-estimator that obtains an estimate $\hat{\xv}$ of $\xv_{0}$ from $\yv$ by solving the convex problem
\begin{equation}
\label{eq:m regularization}
\hat{\xv} := \argmin_{\xv} \mathcal{L}\left(\yv - \Am \xv \right) + \gamma f\left(\xv\right),
\end{equation} 
where the loss function $\mathcal{L}:\mathbb{R}^{m} \to \mathbb{R} $ measures the fit of $\Am\xv$ to the observation vector $\yv$, the penalty function $f:\mathbb{R}^{m} \to \mathbb{R}$ establishes the structure of $\xv$, and $\gamma$ provides a balance between the two functions. Different choices of $\mathcal{L}$ and $f$ leads to different estimators. The most popular among them is the Tikhonov regularization which is given in its simplified form by \cite{tikhonov1977methods}
\begin{equation}
\label{eq:tik-minimization}
\hat{\xv}_{\text{RLS}}  := \argmin_{\xv}  \ ||\yv - \Am \xv ||_2^{2} + \gamma \  ||\xv ||_2^{2}.
\end{equation}
The solution to (\ref{eq:tik-minimization}) is given by the regularized least-square (RLS) estimator
\begin{equation}
\label{eq:R-LS} 
\hat{\xv}_{\text{RLS}} = \left(\Am^{T}\Am+\gamma\Id_{n}\right)^{-1}\Am^{T}\yv,
\end{equation}
where $\Id_{n}$ is ($n \times n $) identity matrix. In general, $\gamma$ is unknown and has to be chosen judiciously.  

On the other hand, several regularization parameter selection methods have been proposed to find the regularization parameter in regularization methods. These include the \emph{generalized cross validation} (GCV) \cite{wahba1990spline}, the \emph{L-curve} \cite{hansen1992analysis, hansen2007adaptive}, and the \emph{quasi-optimal} method \cite{morozov2012methods}. The GCV obtains the regularizer by minimizing the GCV function which suffers from the shortcoming that it may have a very flat minimum that makes it very challenging to be located numerically. The L-curve, on the other hand, is a graphical tool to obtain the regularization parameter which has a very high computational complexity. Finally, the quasi-optimal criterion chooses the regularization parameter without taking into account the noise level. In general, the performance of these methods varies significantly depending on the nature of the problem\footnote{The work presented in this paper is an extended version of \cite{msulimanhybrid}.}. 
\vspace{-10pt}
\subsection{Paper Contributions}
\begin{enumerate}
\item \emph{New regularization approach}: We propose a new approach for linear discrete ill-posed problems that is based on adding an artificial perturbation matrix with a bounded norm to $\Am$. The objective of this artificial perturbation is to improve the challenging singular-value structure of $\Am$. This perturbation affects the fidelity of the model $\yv = \Am \xv_{0} + \zv$, and as a result, the equality relation becomes invalid. We show that using such modification provides a solution with better numerical stability.
\item \emph{New regularization parameter selection method}: We develop a new regularization parameter selection approach that selects the regularizer in a way that minimizes the mean-squared error (MSE) between $\xv_{0}$ and its estimate $\hat{\xv}$, $\mathbb{E} \ ||\hat{\xv} - \xv_{0}||_{2}^{2} $. \footnote{Less work have been done in the literature to provide estimators that are based on the MSE as for example in \cite{eldar2005robust} where the authors derived an estimator for the linear model problem that is based on minimizing the \emph{worst-case} MSE (as opposed to the actual MSE) while imposing a constraint on the unknown vector $\xv_{0}$.}
\item \emph{Generality}: A key feature of the approach is that it does not impose any prior assumptions on $\xv_{0}$. The vector $\xv_{0}$ can be deterministic or stochastic, and in the later case we do not assume any prior statistical knowledge about it. Moreover, we assume that the noise variance $\sigma_{\zv}^{2}$ is unknown. Finally, the approach can be applied to a large number of linear discrete ill-posed problems.
\end{enumerate}
\vspace{-10pt}
\subsection{Paper Organization} 
This paper is organized as follows. Section~\ref{sec:COPRA} presents the formulation of the problem and derive its solution. In Section~\ref{sec:MSE}, we derive the artificial perturbation bound that minimizes the MSE. Further, we derive the proposed approach characteristic equation for obtaining the regularization parameter. Section~\ref{sec:Properties} studies the properties of the approach characteristic equation while Section~\ref{sec:Results} presents the performance of the proposed approach using simulations. Finally, we conclude our work in Section~\ref{sec:conclusion}.
\vspace{-10pt} 
\subsection{Notations} 
Matrices are given in boldface upper case letters (e.g., $\Xm$), column vectors are represented by boldface lower case letters (e.g., $\xv$), and $\left(.\right)^{T}$ stands for the transpose operator. Further, $\mathbb{E}\left(.\right)$, $\Id_{n}$, and $\bm{0}$ denote the expectation operator, the $\left(n \times n \right)$ identity matrix, and the zero matrix, respectively. Notation $||.||_{2}$ refers to the spectral norm for matrices and Euclidean norm for vectors. The operator $\text{diag}\left(.\right)$ returns a vector that contains the diagonal elements of a matrix, and a diagonal matrix if it operates on a vector where the diagonal entries of the matrix are the elements of the vector.
\vspace{-6pt}
\section{Proposed Regularization Approach}
\label{sec:COPRA}
\subsection{Background}
\label{subsec:background}
We consider the linear discrete ill-posed problems in the form $\yv = \Am \xv_{0} + \zv$ and we focus mainly on the case where $m\geq n$ without imposing any assumptions on $\xv_{0}$. The matrix $\Am$ in such problems is ill-conditioned that has a very fast singular values decay \cite{roy2001inverse}. A comparison between the singular values decay of the full rank matrices, the rank deficient matrices, and the ill-posed problems matrices is given in Fig.~\ref{fig:sv decay}. 

From Fig.~\ref{fig:sv decay}, we observe that the singular values of the full column rank matrix are decaying constantly while in the rank definite matrix there is a jump (gap) between the nonzero and the zero singular values. Finally, the singular values of the ill-posed problem matrix are decaying very fast, without a gap, to a significantly small positive number.
\normalsize
\begin{figure}[h!]
    \centering
 \centerline{\includegraphics[width= 3.5in]{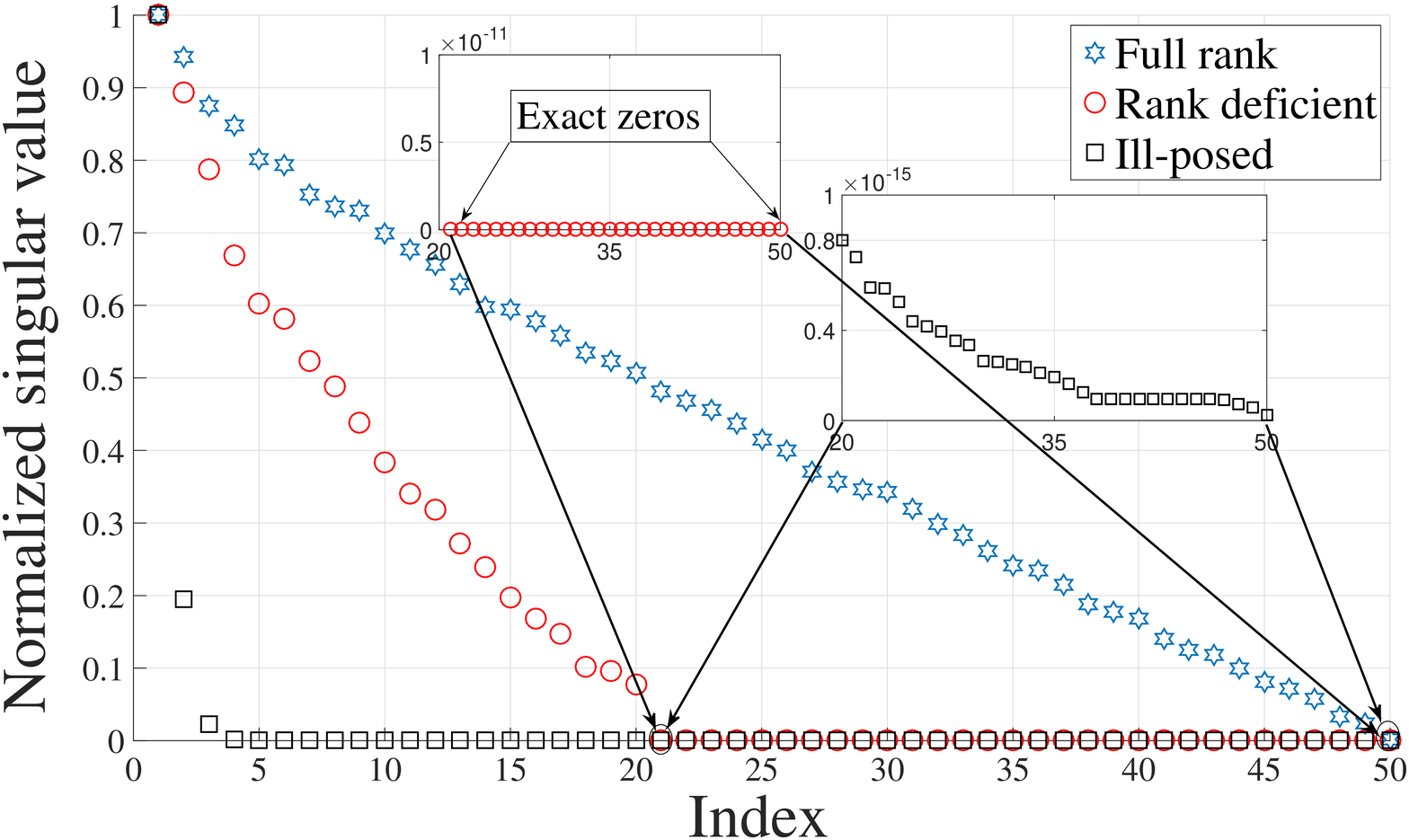}}
\caption{Different singular values decay for different matrices, $\Am \in \mathbb{R}^{50 \times 50}$.}
\label{fig:sv decay}
\end{figure}
\vspace{-10pt}
\subsection{Problem Formulation}
\label{subsec:analysis}
Let us start by considering the LS solution in (\ref{eq:pure LS solution}). In many ill-posed problems, and due to the singular-value structure of $\Am$ and the interaction that it has with the noise, equation~(\ref{eq:pure LS solution}) is not capable of producing a sensible estimate of $\xv_{0}$. Herein, we propose adding an artificial perturbation $\Delta \Am \in \mathbb{R}^{m\times n}$ to $\Am$. We assume that this perturbation, which replaces $\Am$ by $\left(\Am+\Delta \Am\right)$, improves the challenging singular-value structure, and hence is capable of producing a better estimate of $\xv_{0}$. In other words, we assume that using $\left(\Am+\Delta \Am\right)$ in estimating $\xv_{0}$ from $\yv$ provides a better estimation result than using $\Am$. Finally, to provide a balance between improving the singular-value structure and maintaining the fidelity of the basic linear model, we add the constraint $||\Delta \Am||_2 \leq \delta$, $\delta \in \mathbb{R}^{+}$. Therefore, the basic linear model is modified to
\begin{equation}
\label{eq:BDU equation}
\yv \approx \left(\Am+\Delta \Am\right)\xv_{0} + \zv; \  ||\Delta \Am||_2 \leq \delta.
\end{equation}
The question now is what is the best $\Delta \Am$ and the bound on this perturbation. It is clear that these values are important since they affect the model's fidelity and dictate the quality of the estimator. This question is addressed further ahead in this section. For now, let us start by assuming that $\delta$ is known\footnote{We will use this assumption to obtain the proposed estimator solution as a function of $\delta$, then, we will address the problem of obtaining the value of $\delta$.}.

Before proceeding, it worth mentioning that the model in (\ref{eq:BDU equation}) has been considered for signal estimation in the presence of data errors but with strict equality (e.g., \cite{eldar2005robust, el1997robust, chandrasekaran1998parameter}). These methods assume that $\Am$ is not known perfectly due to some error contamination, and that a prior knowledge on the real error bound (which corresponds to $\delta$ in our case) is available. However, in our case the matrix $\Am$ is known perfectly, whereas $\delta$ is unknown.

To obtain an estimate of $\xv_{0}$, we consider minimizing the worst-case residual function of the new perturbed model in (\ref{eq:BDU equation}) which is given by
\begin{eqnarray}
\label{eq:worst-error} %%W\left(\Delta \Am, \hat{\xv}\right)  = 
&\underset{\xv}{\operatorname{\min}} \  \underset{||\Delta \Am||_2 \leq \delta}{\operatorname{\max}}_{} \ Q\left(\xv,\Delta \Am\right) := \ ||\yv - \left(\Am  + \Delta \Am\right) \xv ||_2. 
\end{eqnarray}
\begin{theorem}
The unique minimzer $\hat{\xv}$ of the min-max constrained problem in (\ref{eq:worst-error}) for fixed $\delta > 0$ is given by
\begin{equation}
\label{eq:copra solution} 
\hat{\xv} = \left(\Am^{T}\Am+\rho\left(\delta,\hat{\xv}\right)\Id_{n}\right)^{-1}\Am^{T}\yv,
\end{equation}
where $\rho\left(\delta,\hat{\xv}\right)$ is a regularization parameter that is related to the perturbation bound $\delta$ through
\begin{equation}
\label{eq:sec1}
\rho\left(\delta,\hat{\xv}\right) = \delta \ \frac{ || \yv- \Am \hat{\xv}  ||_2 }{|| \hat{\xv} ||_2}.
\end{equation}
\end{theorem}
\begin{proof}
By using Minkowski inequality \cite{maligranda1995simple}, we find an upper bound for the cost function $Q\left(\xv,\Delta \Am\right)$ in (\ref{eq:worst-error}) as
\begin{align}
\label{eq:bound minko}
%\label{eq:Minkowski inequality}
||\yv - \left(\Am  + \Delta \Am\right) \xv ||_2 & \leq  ||\yv - \Am  \xv ||_2 +  || \Delta \Am  \ \xv  ||_{2} \nonumber\\
%\label{eq:Minkowski inequality2}
&\leq  ||\yv - \Am  \xv ||_2 + || \Delta \Am ||_{2} || \xv ||_{2} \nonumber\\
%\label{eq:Minkowski inequality3} 
&\leq ||\yv - \Am  \xv ||_2 + \delta \ || \xv ||_{2}.
\end{align}
However, upon setting $\Delta \Am$ to be the following rank one matrix
\begin{equation}
\label{eq:perturbtion upper bound}
\Delta \Am = \frac{\left( \Am  \xv -\yv \right) }{ ||\yv - \Am  \xv ||_2} \frac{\xv^{T}}{|| \xv ||_{2}} \delta,
\end{equation}
we can show that the bound in (\ref{eq:bound minko}) is achievable by
\begin{align}
\label{eq:proof of upper bound}
|| \yv - \left(\Am  + \Delta \Am\right) \xv ||_2  & = ||\left(\yv -\Am \xv \right)  +  \frac{\left(\yv - \Am \xv \right)}{||\yv - \Am  \xv  ||_2} \frac{\xv^{T} }{|| \xv ||_{2}}  \xv \delta ||_2 \nonumber\\
 & = ||\left(\yv -\Am \xv \right)  +  \frac{\left(\yv - \Am \xv \right)}{|| \yv- \Am  \xv ||_2} || \xv ||_{2}  \delta||_2. 
\end{align}
Since the two added vectors $\left(\yv -\Am \xv \right)$ and $\frac{\left(\yv - \Am \xv \right)}{|| \yv- \Am  \xv ||_2} || \xv ||_{2}  \delta $ in (\ref{eq:proof of upper bound}) are positively linearly dependent (i.e., pointing in the same direction), we conclude that
\begin{align}
\label{eq:proof of upper bound2}
||\left( \yv  - \Am \xv\right) +  \frac{\left(\yv - \Am \xv \right)}{||  \yv - \Am  \xv ||_2} || \xv ||_{2} \ \delta ||_2  = \underbrace{|| \yv - \Am \xv||_2 + \delta||\xv ||_2}_{W\left(\xv\right)}
\end{align}
As a result, (\ref{eq:worst-error}) can be expressed equivalently by
\begin{equation}
\label{eq:costfunction}
\underset{\xv}{\operatorname{\min}} \  \underset{||\Delta \Am||_2 \leq \delta}{\operatorname{\max}}_{} \ Q\left( \xv, \Delta \Am\right) \equiv \underset{\xv}{\operatorname{\min}} \  W\left(\xv\right).
\end{equation}
Therefore, the solution of (\ref{eq:worst-error}) depends only on $\delta$ and is agnostic to the structure of $\Delta \Am$\footnote{Interestingly, setting the norm of the penalty term in $W\left(\xv\right)$ to be of l1 norm (i.e., $||\xv ||_{1}$) leads to the square-root LASSO \cite{belloni2011square} which is used in sparse signal estimation.}. It is easy to check that the solution space for $W\left(\xv\right)$ is convex in $\xv$, and hence, any local minimum is also a global minimum. But at any local minimum, it either holds that the gradient of $W\left(\xv\right)$ is zero, or $W\left(\xv\right)$ is not differentiable. More precisely, $W\left(\xv\right)$ is not differentiable only at $\xv$ = 0 and when $\yv - \Am \xv= 0$. However, the former case is a trivial case that is not being considered in this paper, while the latter case is not possible by definition. The gradient of $W\left(\xv\right)$ can be obtained as
\begin{align}
\label{eq:gradientI}
&\nabla_{\xv}  W\left(\xv\right)
=  \frac{1}{|| \yv- \Am \xv  ||_2} \Am^{T} \left(\Am \xv - \yv \right) +  \frac{\delta \ \xv}{|| \xv ||_2} \nonumber\\ 
&=  \frac{1}{|| \yv- \Am \xv  ||_2} \left( \Am^{T} \Am \xv  +  \frac{\delta   \ || \yv- \Am \xv  ||_2\ \xv }{|| \xv ||_2}  - \Am^{T}\yv \right).
\end{align}
Solving for $\nabla_{\xv}  W\left(\hat{\xv}\right) = 0$ upon defining $\rho\left(\delta,\hat{\xv}\right)$ as in (\ref{eq:sec1}), we obtain (\ref{eq:copra solution}).
\end{proof}
\begin{remark}
\normalfont
%\subsection{Remarks}
%\begin{enumerate}
The regularization parameter $\rho$ in (\ref{eq:sec1}) is a function of the unknown estimate $\hat{\xv}$, as well as the upper perturbation bound $\delta$ (we have dropped the dependence of $\rho$ on $\delta$ and $\hat{\xv}$ to simplify notation). In addition, it is clear from (\ref{eq:proof of upper bound2}) that $\delta$ controls the weight given to the minimization of the side constraint relative to the minimization of the residual norm. We have assumed that $\delta$ is known to obtain the min-max optimization solution. However, this assumption is not valid in reality. Thus, is it impossible to obtain $\rho$ directly from (\ref{eq:sec1}) given that both $\delta$ and $\hat{\xv}$ are unknowns.
\end{remark}
Now, it is obvious with (\ref{eq:copra solution}) and (\ref{eq:sec1}) in hand, we can eliminate the dependency of $\rho$ on $\hat{\xv}$. By substituting (\ref{eq:copra solution}) in (\ref{eq:sec1}) we obtain after some algebraic manipulations
\begin{align}
\label{eq:secular equation}
&\delta^{2} \Big[\yv^{T}\yv - 2\yv^{T} \Am\left(\Am^{T}\Am+\rho\Id_{n}\right)^{-1}\Am^{T}\yv \nonumber\\
& + ||\Am\left(\Am^{T}\Am+\rho\Id_{n}\right)^{-1}\Am^{T}\yv||^{2} \Big]  \nonumber\\
&= \rho^{2} \yv^{T}\Am\left(\Am^{T}\Am+\rho\Id_{n}\right)^{-2}\Am^{T}\yv.
\end{align}
In this following subsection, we will utilize and simplify (\ref{eq:secular equation}) using some manipulations to obtain $\delta$ that corresponds to an optimal choice of $\rho$ in ill-posed problems.
\vspace{-10pt}
\subsection{Finding the Optimal Perturbation Bound}
\label{choosing optimal bound}
Let us denote the optimal choices of $\rho$ and $\delta$ by $\rho_{\text{o}}$ and $\delta_{\text{o}}$, respectively. To simplify (\ref{eq:secular equation}), we substitute the SVD of $\Am$, then we solve for $\delta^{2}$, and finally we take the trace $\text{Tr}\left(.\right)$ of the two sides considering the evaluation point to be $\left(\delta_{\text{o}},\rho_{\text{o}}\right)$ to get
\begin{align}
\label{eq:secular equation trace}
&\underbrace{\delta_{\text{o}}^2 \ \text{Tr}\left( \left(\Sigmam^2 + \rho_{\text{o}} \Id_{n} \right)^{-2} \Um^{T} \left(\yv \yv^{T}\right) \Um  \right)}_{D\left(\rho_{\text{o}}\right)} \nonumber\\
&=
\underbrace{\text{Tr}\left( \Sigmam^2 \left(\Sigmam^2 + \rho_{\text{o}} \Id_{n} \right)^{-2} \Um^{T} \left(\yv \yv^{T}\right) \Um  \right)}_{N\left(\rho_{\text{o}}\right)}.
\end{align}
In order to obtain a useful expression, let us think of $\delta_{\text{o}}$ as a single universal value that is computed over many realizations of the observation vector $\yv$. Based on this perception, $\yv\yv^{T}$ can be replaced by its expected value $\mathbb{E}(\yv \yv^{T})$. In other words, we are looking for an optimal choice of $\delta$, say $\delta_{\text{o}}$, that is optimal for all realizations of $\yv$. At this point, we assume that such value exists. Then this parameter $\delta_{\text{o}}$ is clearly deterministic. If we sum (\ref{eq:secular equation trace}) for all realizations of $\yv$, and a fixed $\delta_{\text{o}}$, we end replacing $\yv\yv^{T}$ with $\mathbb{E}(\yv \yv^{T})$ which can be expressed using our basic linear model as
\begin{align}
\label{•eq:yy'}
\mathbb{E} \left(\yv \yv^{T} \right)
& =
\Am \Rm_{\xv_{0}} \Am^{T} + \sigma_{\zv}^2 \Id_{m} \nonumber\\
&=
\Um \Sigmam \Vm^{T}\Rm_{\xv_{0}} \Vm \Sigmam \Um^{T}+ \sigma_{\zv}^2 \Id_{m},
\end{align}
where $\Rm_{\xv_{0}} \triangleq \mathbb{E}\left(\xv_{0} \xv_{0}^{T} \right)$ is the covariance matrix of $\xv_{0}$. For a deterministic $\xv_{0}$, $\Rm_{\xv_{0}} = \xv_{0}\xv_{0}^{T}$ is used for notational simplicity.
Substituting (\ref{•eq:yy'}) in both terms of (\ref{eq:secular equation trace}) results
\begin{align}
\label{eq:N term}
N\left(\rho_{\text{o}}\right)&= 
\text{Tr}\left( \Sigmam^2 \left(\Sigmam^2 + \rho_{\text{o}} \Id_{n} \right)^{-2}\Sigmam^{2} \Vm^{T} \Rm_{\xv_{0}} \Vm  \right)\nonumber\\
&+\sigma_{\zv}^{2} \text{Tr}\left(\Sigmam^{2}\left(\Sigmam^{2}+\rho_{\text{o}}\Id_{n}\right)^{-2}\right),
\end{align}
and
\begin{align}
\label{eq:D term}
D\left(\rho_{\text{o}}\right)&= \delta_{\text{o}}^{2}\Big[\text{Tr}\left( \left(\Sigmam^2 + \rho_{\text{o}} \Id_{n} \right)^{-2}\Sigmam^{2} \Vm^{T} \Rm_{\xv_{0}} \Vm  \right) \nonumber\\ 
&+\sigma_{\zv}^{2}\text{Tr}\left(\left(\Sigmam^{2}+\rho_{\text{o}}\Id_{n}\right)^{-2}\right)\Big].
\end{align}
Considering the singular-value structure for the ill-posed problems, we can divide the singular values into two groups of \emph{significant}, or relatively large, and \emph{trivial}, or nearly zero singular value\footnote{This includes the special case when all the singular values are significant and so all are considered.}. As an example, we can see from Fig.~\ref{fig:sv decay} that the singular values of the ill-posed problem matrix are decaying very fast, making it possible to identify the two groups. Based on this, the matrix $\Sigmam$ can be divided into two diagonal sub-matrices, $\Sigmam_{n_{1}}$, which contains the first (significant) $n_{1}$ diagonal entries, and $\Sigmam_{n_{2}}$, which contains the last (trivial) $n_2 = n - n_1$ diagonal entries\footnote{The splitting threshold is obtained as the mean of the eigenvalues multiplied by a certain constant $c$, where c $\in \left(0,1\right)$.}. As a result, $\Sigmam$ can be written as
\begin{equation}
\label{eq:sigma parti}
\Sigmam = 
\begin{bmatrix}

 \Sigmam_{n_{1}}  & \bm{0} \\ 
 \bm{0}& \Sigmam_{n_{2}}
 
\end{bmatrix}.
\end{equation}
Similarly, we can partition $\Vm$ as  $\Vm = [\Vm_{n_{1}}  \  \Vm_{n_{2}}]$ where $\Vm_{n_{1}}\in \mathbb{R}^{n\times n_{1}}$, and $\Vm_{n_{2}}\in \mathbb{R}^{n\times n_{2}}$. Now, we can write $N\left(\rho_{\text{o}}\right)$ in (\ref{eq:N term}) in terms of the partitioned $\Sigmam$ and $\Vm$ as
\begin{align}
\label{eq:numer eqaution 1}
N\left(\rho_{\text{o}}\right)
&=
\text{Tr} \left(\Sigmam_{n_{1}}^2 \left(\Sigmam_{n_{1}}^2 + \rho_{\text{o}} \Id_{n_{1}}  \right)^{-2} \Sigmam_{n_{1}}^2 \Vm_{n_{1}}^{T} \Rm_{\xv_{0}}\Vm_{n_{1}}\right) \nonumber\\
&+
\text{Tr} \left(\Sigmam_{n_{2}}^2 \left(\Sigmam_{n_{2}}^2 + \rho_{\text{o}} \Id_{n_{2}}  \right)^{-2} \Sigmam_{n_{2}}^2 \Vm_{n_{2}}^{T} \Rm_{\xv_{0}}\Vm_{n_{2}}\right)\nonumber\\
&+\sigma_{\zv}^2 \text{Tr}\left(\Sigmam_{n_{1}}^2 \left(\Sigmam_{n_{1}}^2 + \rho_{\text{o}} \Id_{n_{1}} \right)^{-2} \right)\nonumber\\
&+
\sigma_{\zv}^2 \text{Tr}\left(\Sigmam_{n_{2}}^2 \left(\Sigmam_{n_{2}}^2 + \rho_{\text{o}} \Id_{n_{2}} \right)^{-2} \right). 
\end{align}
Given that $\Sigmam_{n_{1}}$ contains the significant singular values and $\Sigmam_{n_{2}}$ contains the relatively small (nearly zero) singular values, we have $\lVert \Sigmam_{n_{2}} \rVert \approx 0$, and so we can approximate $N\left(\rho_{\text{o}}\right)$ as
\begin{align}
\label{eq:numer eqaution 2}
N\left(\rho_{\text{o}}\right)
&\approx
\text{Tr} \left(\Sigmam_{n_{1}}^2 \left(\Sigmam_{n_{1}}^2 + \rho_{\text{o}} \Id_{n_{1}}  \right)^{-2} \Sigmam_{n_{1}}^2 \Vm_{n_{1}}^{T} \Rm_{\xv_{0}}\Vm_{n_{1}}\right)\nonumber\\
&+
\sigma_{\zv}^2 \text{Tr}\left(\Sigmam_{n_{1}}^2 \left(\Sigmam_{n_{1}}^2 + \rho_{\text{o}} \Id_{n_{1}} \right)^{-2} \right) .
\end{align}
Similarly, $D\left(\rho_{\text{o}}\right)$ in (\ref{eq:D term}) can be approximated equivalently as
\begin{align}
\label{eq:denumer approximation}
&D \left(\rho_{\text{o}}\right) \approx
\sigma_{\zv}^2 \  \text{Tr}\left(\left(\Sigmam_{n_{1}}^2 + \rho_{\text{o}} \Id_{n_{1}} \right)^{-2} \right) + \frac{ n_{2} \sigma_{\zv}^2}{\rho_{\text{o}}^{2}} \nonumber\\
&+  \text{Tr} \left( \left(\Sigmam_{n_{1}}^{2} + \rho_{\text{o}} \Id_{n_{1}}  \right)^{-2} \Sigmam_{n_{1}}^{2} \Vm_{n_{1}}^{T} \Rm_{\xv_{0}}\Vm_{n_{1}}\right).
\end{align}
Substituting \eqref{eq:numer eqaution 2} and \eqref{eq:denumer approximation} in (\ref{eq:secular equation trace}) and manipulating, we obtain 
\begin{align}
\label{eq:eta min 2}
\delta_{\text{o}}^2 &\approx
\Big[\sigma_{\zv}^2 \text{Tr}\left(\Sigmam_{n_{1}}^2 \left(\Sigmam_{n_{1}}^2 + \rho_{\text{o}} \Id_{n_{1}} \right)^{-2} \right)\nonumber\\
&+
\text{Tr} \left(\Sigmam_{n_{1}}^2 \left(\Sigmam_{n_{1}}^2 + \rho_{\text{o}} \Id_{n_{1}}  \right)^{-2} \Sigmam_{n_{1}}^2 \Vm_{n_{1}}^{T} \Rm_{\xv_{0}}\Vm_{n_{1}}\right)\Big] \Big/ \nonumber\\
&\Big[\sigma_{\zv}^2 \text{Tr}\left( \left(\Sigmam_{n_{1}}^2 + \rho_{\text{o}} \Id_{n_{1}} \right)^{-2} \right)
+ \frac{n_{2} \sigma_{\zv}^{2}
}{\rho_{\text{o}}^{2}} \nonumber\\
&+ \text{Tr} \left( \left(\Sigmam_{n_{1}}^{2} + \rho_{\text{o}} \Id_{n_{1}}  \right)^{-2} \Sigmam_{n_{1}}^{2} \Vm_{n_{1}}^{T} \Rm_{\xv_{0}}\Vm_{n_{1}}\right)\Big].
\end{align}
The bound $\delta_{\text{o}}$ in (\ref{eq:eta min 2}) is a function of the unknowns $\rho_{\text{o}}$, $\sigma_{\zv}^{2}$, and $\Rm_{\xv_{0}}$. In fact, estimating $\sigma_{\zv}^{2}$ and $\Rm_{\xv_{0}}$ without any prior knowledge is a very tedious process. The problem becomes worse when $\xv_{0}$ is deterministic. In such case, the exact value of $\xv_{0}$ is required to obtain $\Rm_{\xv_{0}} = \xv_{0}\xv_{0}^{T}$. In the following section, we will use the MSE as a criterion to eliminate the dependence of $\delta_{\text{o}}$ on these unknowns and a result to set the value of the perturbation bound.
\section{Minimizing the MSE for the solution of the proposed perturbation approach}
\label{sec:MSE}
The MSE for an estimate $\hat{\xv}$ of $\xv_{0}$ is given by
\begin{equation}
\label{eq:MSE}
\text{MSE} =\mathbb{E}\big[ ||\hat{\xv} - \xv_{0}||^{2}  \big] =\text{Tr}\left( \mathbb{E}\left( (\hat{\xv} - \xv_{0}) (\hat{\xv} - \xv_{0})^{T}  \right) \right).
\end{equation}
Since the solution of the proposed approach problem in (\ref{eq:worst-error}) is given by (\ref{eq:copra solution}), we can substitute for $\hat{\xv}$ from (\ref{eq:copra solution}) in (\ref{eq:MSE}) and then use the SVD of $\Am$ to obtain
\begin{align}
\label{eq:MSE2}
\text{MSE}\left(\rho\right)
&=
\sigma_{\zv}^{2} \text{Tr}\left(\Sigmam^{2} \left(\Sigmam^{2} + \rho\Id_{n} \right)^{-2} \right) \nonumber \\
&+
\rho^{2} \text{Tr}\left( \left(\Sigmam^2 + \rho\Id_{n} \right)^{-2}\Vm^{T}\Rm_{\xv_{0}}\Vm \right).
\end{align}
\begin{theorem}
For $\sigma_{\zv}^{2} >0$, the approximate value for the optimal regularizer $\rho_{\text{o}}$ of (\ref{eq:MSE2}) that approximately minimizes the MSE is given by
\begin{equation}
\label{eq:gamma min approx}
\rho_{\text{o}}\approx \frac{ \sigma_{\zv}^2}{\text{Tr}\left(\Rm_{\xv_{0}}\right) /n}.
\end{equation}
\end{theorem}
\begin{proof}
We can easily prove that the function in (\ref{eq:MSE2}) is convex in $\rho$, and hence its global minimizer (i.e., $\rho_{\text{o}}$) can be obtained by differentiating (\ref{eq:MSE2}) with respect to $\rho$ and setting the result to zero, i.e.,
\begin{align}
\label{eq:MSE'}
&\nabla_{\rho} \ \text{MSE}\left(\rho\right)
=
-2\sigma_{\zv}^2 \text{Tr}\left(\Sigmam^2 \left(\Sigmam^2 + \rho\Id_{n} \right)^{-3} \right) \nonumber \\
&+
2 \rho \underbrace{\text{Tr}\left( \Sigmam^{2} \left(\Sigmam^2 + \rho\Id_{n} \right)^{-3}\Vm^{T}\Rm_{\xv_{0}}\Vm \right)}_{S} = 0.
\end{align}
Equation (\ref{eq:MSE'}) dictates the relationship between the optimal regularization parameter and the problem parameters. By solving (\ref{eq:MSE'}), we can obtain the optimal regularizer $\rho_{\text{o}}$. However, in the general case, and with lack of knowledge about $\Rm_{\xv_{0}}$, we cannot obtain a closed-form expression for $\rho_{\text{o}}$. As a result, we will seek to obtain a suboptimal regularizer in the MSE sense that minimizes (\ref{eq:MSE2}) approximately. In what follows, we show how through some bounds and approximations, we can obtain this suboptimal regularizer. 

By using the trace inequalities in [\cite{wang1986trace}, eq.(5)], we can bound the second term in (\ref{eq:MSE'}) by
\begin{align}
\label{eq:inequality}
&\lambda_{\text{min}}\left(\Rm_{\xv_{0}}\right)\text{Tr}\left(\Sigmam^2 \left(\Sigmam^2 + \rho\Id_{n} \right)^{-3} \right) \nonumber\\
& \leq S=  \text{Tr}\left( \Sigmam^{2} \left(\Sigmam^2 + \rho\Id_{n} \right)^{-3}\Vm^{T}\Rm_{\xv_{0}}\Vm \right) \nonumber\\
& \leq \lambda_{\text{max}}\left(\Rm_{\xv_{0}}\right)\text{Tr}\left(\Sigmam^2 \left(\Sigmam^2 + \rho\Id_{n} \right)^{-3} \right),
\end{align}
where $\lambda_{i}$ is the $i$'th eigenvalue of $\Rm_{\xv_{0}}$. Our main goal in this paper is to find a solution that is approximately feasible for all discrete ill-posed problems and also suboptimal in some sense. In other words, we would like to find a $\rho_{\text{o}}$, for all (or almost all) possible $\Am$, that minimizes the MSE approximately. To achieve this, we consider an \emph{average} value of $S$ based on the inequalities in (\ref{eq:inequality}) as our evaluation point, i.e., 
\begin{align}
\label{eq:d}
S \approx \text{Tr}\left(\Sigmam^2 \left(\Sigmam^2 + \rho\Id_{n} \right)^{-3} \right) \frac{\text{Tr}\left(\Rm_{\xv_{0}} \right)}{n}.
\end{align}
Substituting (\ref{eq:d}) in (\ref{eq:MSE'}) yields\footnote{Another way to look at (\ref{eq:d}) is that we can replace $\Vm^{T} \Rm_{\xv_{0}} \Vm $ inside the trace in (\ref{eq:MSE'}) by a diagonal matrix $\Fm = \text{diag}\left(\text{diag} \left(\Vm^{T} \Rm_{\xv_{0}} \Vm \right)\right)$ without affecting the result. Then, we replace this positive definite matrix $\Fm$ by an identity matrix multiplied by scalar which is given by the average value of the diagonal entries of $\Fm$.}
\begin{align}
\label{eq:MSE' approx}
&\nabla_{\rho} \ \text{MSE}\left(\rho\right)
\approx
-2\sigma_{\zv}^2 \text{Tr}\left(\Sigmam^2 \left(\Sigmam^2 + \rho\Id_{n} \right)^{-3} \right) \nonumber \\
&+
2 \rho \frac{\text{Tr}\left(\Rm_{\xv_{0}} \right)}{n}\text{Tr}\left(\Sigmam^2 \left(\Sigmam^2 + \rho\Id_{n} \right)^{-3} \right) = 0.
\end{align}
Note that the same approximation can be applied from the beginning to the second term in (\ref{eq:MSE2}) and the same result in (\ref{eq:MSE' approx}) will be obtained after taking the derivative of the new approximated MSE function. In Appendix~\ref{Apen error}, we provide the error analysis for this approximation and show that it is bounded in very small feasible region.\\ 
Equation (\ref{eq:MSE' approx}) can now be solved to obtain $\rho_{\text{o}}$ as in (\ref{eq:gamma min approx}).\footnote{In fact, one can prove that when $\xv_{0}$ is i.i.d., (\ref{eq:MSE' approx}) and (\ref{eq:MSE'}) are exactly equivalent to each other (see Appendix~\ref{Apen error}).}
\begin{figure*}[t!]
\setcounter{equation}{36}
\begin{align}
\label{eq:IBPRfunction}
G\left(\rho_{\text{o}}\right) &= \underbrace{\text{Tr}\left( \Sigmam^2 \left(\Sigmam^2 + \rho_{\text{o}} \Id_{n} \right)^{-2} \Um^{T}\yv\yv^{T}\Um \right)\text{Tr}\left( \left(\Sigmam_{n_{1}}^2 + \rho_{\text{o}} \Id_{n_{1}} \right)^{-2}\left(\beta\Sigmam_{n_{1}}^2 + \rho_{\text{o}} \Id_{n_{1}}\right) \right) + \frac{n_{2}}{\rho_{\text{o}}}\text{Tr}\left( \Sigmam^2 \left(\Sigmam^2 + \rho_{\text{o}} \Id_{n} \right)^{-2}\Um^{T}\yv\yv^{T}\Um \right)}_{G_{1}\left(\rho_{\text{o}}\right)}\nonumber\\
&- \underbrace{\text{Tr}\left(\left(\Sigmam^2 + \rho_{\text{o}} \Id_{n} \right)^{-2} \Um^{T}\yv\yv^{T}\Um\right)\text{Tr}\left( \Sigmam_{n_{1}}^2 \left(\Sigmam_{n_{1}}^2 + \rho_{\text{o}} \Id_{n_{1}} \right)^{-2} \left(\beta\Sigmam_{n_{1}}^2 + \rho_{\text{o}} \Id_{n_{1}}\right) \right)}_{ G_{2}\left(\rho_{\text{o}}\right)} = 0.
\end{align}
\hrulefill
\end{figure*}
\end{proof}
\begin{remark}
\normalfont
%\textbf{Remark 1.}
The solution in (\ref{eq:gamma min approx}) shows that there always exists a positive $\rho_{\text{o}}$, for $\sigma_{\zv}^{2} \neq 0$, which approximately minimizes the MSE in (\ref{eq:MSE2}). The conclusion that the regularization parameter is generally dependent on the noise variance has been shown before in different contexts (see for example \cite{zachariah2015online, hemmerle1975explicit}). For the special case where the entries of $\xv_{0}$ are independent and identically distributed (i.i.d.) with zero mean, we have $\Rm_{\xv_{0}} = \sigma_{\xv_{0}}^{2}\Id_{n}$. Since the linear minimum mean-squared error (LMMSE) estimator of $\xv_{0}$ in $\yv = \Am \xv_{0} + \zv$ is defined as \cite{kay2013fundamentals}
\setcounter{equation}{33}
\begin{equation}
\label{eq:lmmse2}
\hat{\xv}_{\text{LMMSE}}=  \left(\Am^{T} \Am + \sigma_{\zv}^{2} \Rm_{\xv_{0}}^{-1} \Id_{n}\right)^{-1}  \Am^{T} \yv,
\end{equation}
substituting $\Rm_{\xv_{0}} = \sigma_{\xv_{0}}^{2}\Id$ makes the LMMSE regularizer in (\ref{eq:lmmse2}) equivalent to $\rho_{\text{o}}$ in (\ref{eq:gamma min approx}) since $\rho_{\text{o}}= \frac{\sigma_{\zv}^{2}}{\text{Tr}\left(\Rm_{\xv_{0}}\right)/n} = \frac{\sigma_{\zv}^{2}}{\sigma_{\xv_{0}}^{2}}$. This shows that (\ref{eq:gamma min approx}) is exact when the input is white, while for a general input $\xv_{0}$, the optimum matrix regularizer is given in (\ref{eq:lmmse2}). In other words, the result in (\ref{eq:gamma min approx}) provides an approximate optimum scalar regularizer for a general colored input. Note that since $ \sigma_{\zv}^2$ and $\Rm_{\xv_{0}}$ are unknowns, $\rho_{\text{o}}$ cannot be obtained directly from (\ref{eq:gamma min approx}). 
\end{remark}
We are now ready in the following subsection to use the result in (\ref{eq:gamma min approx}) along with the perturbation bound expression in (\ref{eq:eta min 2}) and some reasonable manipulations and approximations to eliminate the dependency of $\delta_{\text{o}}$ in (\ref{eq:eta min 2}) on the unknowns $\sigma_{\zv}^{2}$, and $\Rm_{\xv_{0}}$. Then, we will select a pair of $\delta_{\text{o}}$ and $\rho_{\text{o}}$ from the space of all possible values of $\delta$ and $\rho$ that minimizes the MSE of the proposed estimator solution.
\vspace{-10pt}
\subsection{Setting the Optimal Perturbation Bound that Minimizes the MSE}
We start by applying the same reasoning leading to (\ref{eq:d}) for both the numerator and the denominator of (\ref{eq:eta min 2}) and manipulate to obtain
\begin{align}
\label{eq:eta min 3}
&{\text{Tr}\left( \Sigmam_{n_{1}}^2 \left(\Sigmam_{n_{1}}^2 + \rho_{\text{o}} \Id_{n_{1}} \right)^{-2} \left(\Sigmam_{n_{1}}^2 + \frac{n_1 \sigma_{\zv}^2}{\text{Tr}\left(\Rm_{\xv_{0}}\right)} \Id_{n_{1}} \right)  \right)} \nonumber\\
&\approx 
\delta_{\text{o}}^2 \Bigg[ \text{Tr}\left(\left(\Sigmam_{n_{1}}^2 + \rho_{\text{o}} \Id_{n_{1}} \right)^{-2} \left(\Sigmam_{n_{1}}^2 + \frac{n_1 \sigma_{\zv}^2}{\text{Tr}\left(\Rm_{\xv_{0}}\right)} \Id_{n_{1}} \right)  \right) \nonumber\\
& + \frac{n_2 n_1 \sigma_{\zv}^2
}{\rho_{\text{o}}^{2}{\text{Tr}\left(\Rm_{\xv_{0}}\right)}} \Bigg].
\end{align}
In Section~\ref{sec:Results}, we verify this approximation using simulations.

Now, we will use the relationship of $\sigma_{\zv}^{2}$ and $\text{Tr}\left(\Rm_{\xv_{0}}\right)$ in (\ref{eq:gamma min approx}) to the suboptimal regularizer $\frac{n_{1}}{n}\rho_{\text{o}} \approx \frac{n_{1} \sigma_{\zv}^{2}} {\text{Tr}(\Rm_{\xv_{0}})}$ to impose a constrain on (\ref{eq:eta min 3}) that makes the selected perturbation bound minimizes the MSE and as a result to make (\ref{eq:eta min 3}) an implicit equation in $\delta_{\text{o}}$ and $\rho_{\text{o}}$ only. By doing this, we obtain after some algebraic manipulations
\begin{align}
\label{eq:eta min final 1}
\delta_{\text{o}}^2 \approx
\frac
{\text{Tr}\left( \Sigmam_{n_{1}}^{2} \left(\Sigmam_{n_{1}}^2 + \rho_{\text{o}} \Id_{n_{1}} \right)^{-2} \left(\frac{n}{n_1}\Sigmam_{n_{1}}^2 + \rho_{\text{o}} \Id_{n_{1}} \right)  \right)}
{\text{Tr}\left( \left(\Sigmam_{n_{1}}^2 + \rho_{\text{o}} \Id_{n_{1}} \right)^{-2}  \left(\frac{n}{n_1}\Sigmam_{n_{1}}^2 + \rho_{\text{o}} \Id_{n_{1}} \right) \right) +\frac{n_{2}}{\rho_{\text{o}}}}.
\end{align}
The expression in (\ref{eq:eta min final 1}) reveals that any $\delta_{\text{o}}$ satisfying (\ref{eq:eta min final 1}) minimizes the MSE approximately. Now, we have two equations (\ref{eq:secular equation}) (evaluated at $\delta_{0}$ and $\rho_{0}$) and (\ref{eq:eta min final 1}) in two unknowns $\delta_{0}$ and $\rho_{0}$. Solving these equations and then applying the SVD of $\Am$ to the result equation, result in the characteristic equation for the proposed constrained perturbation regularization approach (COPRA) in (\ref{eq:IBPRfunction}), where $\beta = \frac{n}{n_1}$.  

The COPRA characteristic equation in (\ref{eq:IBPRfunction}) is a function of the problem matrix $\Am$, the received signal $\yv$, and regularization parameter $\rho_{\text{o}}$ which is the only unknown in (\ref{eq:IBPRfunction}. Solving for $G\left(\rho_{\text{o}}\right)= 0$ should lead to the regularization parameter $\rho_{\text{o}}$ that approximately minimizes the MSE of the estimator. Our main interest then is to find a positive root $\rho^{*}_{\text{o}} > 0$ for (\ref{eq:IBPRfunction}). In the following section, we study the main properties of this equation and examine the existence and uniqueness of its positive root. Before that, it worth mentioning the following remark
\setcounter{equation}{37}
\begin{remark}
\normalfont
%\label{p6}
A special case of the proposed COPRA approach is when all the singular values are significant and so no truncation is required (full column rank matrix, see Fig.~\ref{fig:sv decay}). This is the case where $n_{1}=n$, $n_{2}=0$, and $\Sigmam_{n_{1}}=\Sigmam$. Substituting these values in (\ref{eq:IBPRfunction}) we obtain
\begin{align}
\label{eq:IBPRfunction specail}
&\bar{G}\left(\rho_{\text{o}}\right) = \nonumber\\
&\text{Tr}\left( \Sigmam^2 \left(\Sigmam^2 + \rho_{\text{o}} \Id_{n} \right)^{-2} \Um^{T}\yv\yv^{T}\Um \right)\text{Tr}\left( \left(\Sigmam^{2} + \rho_{\text{o}} \Id_{n} \right)^{-1} \right) \nonumber\\
&- \text{Tr}\left(\left(\Sigmam^{2} + \rho_{\text{o}} \Id_{n} \right)^{-2} \Um^{T}\yv\yv^{T}\Um\right)\text{Tr}\left( \Sigmam^{2} \left(\Sigmam^2 + \rho_{\text{o}} \Id_{n} \right)^{-1}\right)\nonumber\\
&= 0.
\end{align}
\end{remark}
%\begin{remark}
%\normalfont
%In the free noise scenario (i.e., $\sigma_{\zv}^{2} \to 0$ ), the expression in (\ref{eq:IBPRfunction}) simplified to
%\begin{align}
%\label{eq:IBPRfunction specail2}
%&\lim_{\sigma_{\zv}^{2} \to 0} G\left(\rho_{\text{o}}\right) = 
%\text{Tr}\left( \Sigmam^2 \left(\Sigmam^2 + \rho_{\text{o}} \Id_{n} \right)^{-2} \Um^{T}\yv\yv^{T}\Um \right)\times \nonumber\\
%&\text{Tr}\left(\Sigmam^{2}_{n1} \left(\Sigmam_{n1}^{2} + \rho_{\text{o}} \Id_{n} \right)^{-2} \right) - \text{Tr}\left(\left(\Sigmam^{2} + \rho_{\text{o}} \Id_{n} \right)^{-2} \Um^{T}\yv\yv^{T}\Um\right) \nonumber\\
%&\times\text{Tr}\left( \Sigmam_{n1}^{4} \left(\Sigmam^2 + \rho_{\text{o}} \Id_{n} \right)^{-2}\right)
%\end{align}
%which provide a significant computational complexity reduction as shown in the next section.
%\end{remark}
\vspace{6pt}
\section{Analysis of the function $G\left(\rho_{\text{o}}\right)$}
\label{sec:Properties}
In this section, we provide a detailed analysis for the COPRA function $G\left(\rho_{\text{o}}\right)$ in (\ref{eq:IBPRfunction}). We start by examining some main properties of $G\left(\rho_{\text{o}}\right)$ that are straightforward to proof.
\begin{property}
\label{p4}
$G \left(\rho_{\text{o}}\right)$ is continuous over the interval $\left(0, +\infty\right)$.
\end{property}

\begin{property}
\label{p44}
$G \left(\rho_{\text{o}}\right)$ has $n$ different discontinuities at $\rho_{\text{o}} = -\sigma_{i}^{2}, \forall i \in [1, n]$. However, these discontinuities are of no interest as far as COPRA is considered.
\end{property}

\begin{property}
\label{p1}
$\lim_{\rho_{\text{o}} \to 0^{+}} G\left(\rho_{\text{o}}\right)  = +\infty$.
\end{property}

\begin{property}
\label{p2}  
$\lim_{\rho_{\text{o}} \to 0^{-}} G\left(\rho_{\text{o}}\right)  = -\infty$.
\end{property}

\begin{property}
\label{p3}
$\lim_{\rho_{\text{o}} \to +\infty} G\left(\rho_{\text{o}}\right) = 0$.
\end{property}

Property \ref{p1} and \ref{p2} show clearly that $G\left(\rho_{\text{o}}\right)$ has a discontinuity at $\rho_{\text{o}} = 0$.

\begin{property}
\label{p5}
Each of the functions $G_{1}\left(\rho_{\text{o}}\right)$ and $G_{2}\left(\rho_{\text{o}}\right)$ in (\ref{eq:IBPRfunction}) is completely monotonic in the interval $(0 ,+\infty)$.
\end{property}
\begin{proof}
According to \cite{feller2008introduction, widder2015laplace}, a function $F\left(\rho_{\text{o}}\right)$ is completely monotonic if it satisfies
\begin{align}
\label{eq:completeMonotoneCond}
\left(-1\right)^{n} F^{\left( n \right)}\left(\rho_{\text{o}}\right) \geq 0, \ 0 <\rho_{\text{o}} < \infty ,  \forall n\in \mathbb{N},
\end{align}
where $F^{(n)}\left(\rho_{\text{o}}\right)$ is the $n$'th derivative of $F\left(\rho_{\text{o}}\right)$.\\
By continuously differentiating $G_{1}\left(\rho_{\text{o}}\right)$ and $G_{2}\left(\rho_{\text{o}}\right)$, we can easily show that both functions satisfy the monotonic condition in (\ref{eq:completeMonotoneCond}).
\end{proof}
\begin{theorem}
\label{th1}
The COPRA function $G\left(\rho_{\text{o}}\right)$ in (\ref{eq:IBPRfunction}) has at most two roots in the interval $\left(0, +\infty \right).$
\end{theorem}
\begin{proof}
The proof of Theorem~\ref{th1} will be conducted in two steps. Firstly, it has been proved in \cite{kammler1976chebyshev, kammler1979least} that any completely monotonic function can be approximated as a sum of exponential functions. That is, if $F\left(\rho_{\text{o}}\right)$ is a completely monotonic, it can be approximated as
\begin{align}
\label{eq:completeMonotoneApprox}
F\left(\rho_{\text{o}}\right) \approx \sum_{i=1}^{l} l_i e^{-k_{i} \rho_{\text{o}}},
\end{align}
where $l$ is the number of the terms in the sum and $l_{i}$ and $k_{i}$ are two constants. It is shown that there always exists a best uniform approximation to $F\left(\rho_{\text{o}}\right)$ where the error in this approximation gets smaller as we increase the number of the terms $l$. However, our main concern here is the relation given by (\ref{eq:completeMonotoneApprox}) more than finding the best number of the terms or the unknown parameters $l_{i}$ and $k_{i}$. To conclude, both functions $G_{1}\left(\rho_{\text{o}}\right)$ and $G_{2}\left(\rho_{\text{o}}\right)$ in (\ref{eq:IBPRfunction}) can be approximated by a sum of exponential functions.

Secondly, it is shown in \cite{shestopaloff2008sums} that the sum of exponential functions has at most two intersections with the abscissa. Consequently, and since the relation in (\ref{eq:IBPRfunction}) can be expressed as a sum of exponential functions, the function $G\left({\rho}_{\text{o}}\right)$ has at most two roots in the interval $\left(0,+\infty \right)$. 
\end{proof}
\begin{theorem}
\label{th2}
There exists a sufficiently small positive value $\epsilon$, 
such that $\epsilon \to 0^{+}$ and $\epsilon \ll \sigma_{i}^{2}$, $ \forall i \in [1,n]$ where the value of the COPRA function $G\left(\rho_{\text{o}}\right)$ in (\ref{eq:IBPRfunction}) is zero (i.e., $\epsilon$ is a positive root for (\ref{eq:IBPRfunction})). However, we are not interested in this root in the proposed COPRA. 
\end{theorem}
\begin{proof}
The proof of Theorem~\ref{th2} is in Appendix~\ref{Apen A}.
\end{proof}
\begin{theorem}
\label{th3}
A sufficient condition for the function $G\left(\rho_{\text{o}}\right)$ to approach zero at $\rho_{\text{o}} = +\infty $ from a positive direction is given by
\begin{equation}
\label{eq:solutioncondition}
n  \text{Tr}\left(\Sigmam^{2} \bv\bv^{T}\right)
> \text{Tr}\left(\Sigmam_{n_{1}}^{2}\right) \text{Tr}\left(\bv\bv^{T}\right)
\end{equation}
where $\bv = \Um^{T}\yv.$
\end{theorem}
\begin{proof} 
Let $\bv = \Um^{T}\yv$ as in (\ref{eq:IBPRfunction}). Given that $\Sigmam^{2}$ is a diagonal matrix, $ \Sigmam^{2}= \text{diag}\left(\sigma_{1}^{2}, \sigma_{2}^{2}, \dotsi,\sigma_{n}^{2}\right)$, and from the trace function property, we can replace $\bv \bv^{T} = \Um^{T}\yv\yv^{T}\Um$ in (\ref{eq:IBPRfunction}) by a diagonal matrix $\bv\bv^{T}_{d}$ that contains $\bv\bv^{T}$ diagonal entries without affecting the result. By defining $\bv \bv^{T}_{d}= \text{diag}\left(b_{1}^{2}, b_{2}^{2}, \dotsi,b_{n}^{2}\right)$, we can write (\ref{eq:IBPRfunction}) as 
\begin{align}
\label{eq2:MSE B sum 1}
&G\left(\rho_{\text{o}}\right)
=
\frac{\beta}{\rho_{\text{o}}^{4}}\sum_{j=1}^{n}\frac{\sigma_{j}^{2} b_{j}^{2}}{\left(\frac{\sigma_{j}^{2}}{\rho_{\text{o}}}+1\right)^{2}}\sum_{i=1}^{n_{1}}\frac{ \sigma_{i}^{2}}{\left(\frac{\sigma_{i}^{2}}{\rho_{\text{o}}}+1\right)^{2}} \nonumber\\
&+
\frac{1}{\rho_{\text{o}}^{3}}\sum_{j=1}^{n}\frac{\sigma_{j}^{2} b_{j}^{2}}{\left(\frac{\sigma_{j}^{2}}{\rho_{\text{o}}}+1\right)^{2}}\sum_{i=1}^{n_{1}}\frac{1}{\left(\frac{\sigma_{i}^{2}}{\rho_{\text{o}}}+1\right)^{2}} \nonumber\\
&-
\frac{\beta}{\rho_{\text{o}}^{4}}\sum_{j=1}^{n}\frac{b_{j}^{2}}{\left(\frac{\sigma_{j}^{2}}{\rho_{\text{o}}}+1\right)^{2}}\sum_{i=1}^{n_{1}}\frac{\sigma_{i}^{4}}{\left(\frac{\sigma_{i}^{2}}{\rho_{\text{o}}}+1\right)^{2}}\nonumber\\
&-
\frac{1}{\rho_{\text{o}}^{3}}\sum_{j=1}^{n}\frac{b_{j}^{2}}{\left(\frac{\sigma_{j}^{2}}{\rho_{\text{o}}}+1\right)^{2}}\sum_{i=1}^{n_{1}}\frac{\sigma_{i}^{2}}{\left(\frac{\sigma_{i}^{2}}{\rho_{\text{o}}}+1\right)^{2}} \nonumber\\
&+
\frac{n_{2}}{\rho_{\text{o}}^{3}}\sum_{j=1}^{n}\frac{\sigma_{j}^{2}b_{j}^{2}}{\left(\frac{\sigma_{j}^{2}}{\rho_{\text{o}}}+1\right)^{2}}.
\end{align}
Then, we use some algebraic manipulations to obtain
\begin{align}
\label{eq2:MSE B sum 2}
 &G\left(\rho_{\text{o}}\right)
=
\frac{1}{\rho_{\text{o}}^{3}}\sum_{j=1}^{n}\frac{\sigma_{j}^{2} b_{j}^{2}}{\left(\frac{\sigma_{j}^{2}}{\rho_{\text{o}}}+1\right)^{2}} \Bigg[ \frac{\beta}{\rho_{\text{o}}}\sum_{i=1}^{n_{1}}\frac{\sigma_{i}^{2}}{\left(\frac{\sigma_{i}^{2}}{\rho_{\text{o}}}+1\right)^{2}} \nonumber\\
&+
\sum_{i=1}^{n_{1}}\frac{1}{\left(\frac{\sigma_{i}^{2}}{\rho_{\text{o}}}+1\right)^{2}} + n_{2} \Bigg]-
\frac{1}{\rho_{\text{o}}^{3}}\sum_{j=1}^{n}\frac{b_{j}^{2}}{\left(\frac{\sigma_{j}^{2}}{\rho_{\text{o}}}+1\right)^{2}} \times \nonumber\\
&\Bigg[ \frac{\beta}{\rho_{\text{o}}}\sum_{i=1}^{n_{1}}\frac{\sigma_{i}^{4}}{\left(\frac{\sigma_{i}^{2}}{\rho_{\text{o}}}+1\right)^{2}} +
\sum_{i=1}^{n_{1}}\frac{\sigma_{i}^{2}}{\left(\frac{\sigma_{i}^{2}}{\rho_{\text{o}}}+1\right)^{2}}\Bigg].
\end{align}
Now, evaluating the limit of (\ref{eq2:MSE B sum 2}) as $\rho_{\text{o}} \to +\infty$ we obtain 
\begin{align}
\label{eq2:MSE B sum 3}
\lim_{\rho_{\text{o}} \to +\infty} G\left(\rho_{\text{o}}\right)
&=
\left(\lim_{\rho_{\text{o}} \to +\infty} \frac{1}{\rho_{\text{o}}^{3}}\right) \nonumber\\
&\times \Bigg[\sum_{j=1}^{n}\sigma_{j}^{2} b_{j}^{2} \Big( \tau \beta\sum_{i=1}^{n_{1}}\sigma_{i}^{2} + \sum_{i=1}^{n_{1}}1 + n_{2} \Big)\nonumber\\
&-
\sum_{j=1}^{n}b_{j}^{2} \Big( \tau \beta\sum_{i=1}^{n_{1}}\sigma_{i}^{4} +\sum_{i=1}^{n_{1}}\sigma_{i}^{2}\Big)\Bigg],
\end{align}
where $\tau = \lim_{\rho_{\text{o}} \to +\infty} \frac{1}{\rho_{\text{o}}}$. The relation in (\ref{eq2:MSE B sum 3}) can be simplified to
\begin{align}
\label{eq2:MSE B sum 4}
&\lim_{\rho_{\text{o}} \to +\infty} G\left(\rho_{\text{o}}\right) = \left(\lim_{\rho_{\text{o}} \to +\infty} \frac{1}{\rho_{\text{o}}^{3}}\right) \nonumber\\
&\times \Bigg[\sum_{j=1}^{n}\sigma_{j}^{2} b_{j}^{2} \Big( \tau \beta\sum_{i=1}^{n_{1}}\sigma_{i}^{2} + n \Big)- \sum_{j=1}^{n}b_{j}^{2} \Big( \tau \beta\sum_{i=1}^{n_{1}}\sigma_{i}^{4} +\sum_{i=1}^{n_{1}}\sigma_{i}^{2}\Big)\Bigg]
\end{align}
It is obvious that the limit in (\ref{eq2:MSE B sum 4}) is zero. However, the direction where the limit approaches zero depends on the sign of the term between the square brackets. For $G\left(\rho_{\text{o}}\right)$ to approach zero from the positive direction, knowing that the terms that are independent of $\tau$ are the dominants, the following condition must hold:
\begin{equation}
\label{eq2:MSE B sum 5}
n \left(\sum_{j=1}^{n} \sigma_{j}^2 b_{j}^{2} \right)
> \left(\sum_{i=1}^{n_{1}} \sigma_{i}^2 \right)  \left(\sum_{j=1}^{n} b_{j}^{2}\right).
\end{equation}
Which is the same as (\ref{eq:solutioncondition}).
\end{proof}
%\begin{figure*}[t!]\centering
%	\setlength\figureheight{7.5cm}
%	\setlength\figurewidth{15cm}
%	\input{23a7a.tikz}
%	\caption{The decay of singular values of the tested ill-posed problems.}
%	\label{fig2:singular values decay}
%\end{figure*}
\begin{theorem}
\label{theo2:th4}
If (\ref{eq:solutioncondition}) is satisfied, then $G\left(\rho_{\text{o}}\right)$ has a unique positive root in the interval $\left(\epsilon, +\infty \right)$.
\end{theorem}
\begin{proof}
According to Theorem~\ref{th1}, the function $G\left(\rho_{\text{o}}\right)$ can have no root, one, or two roots. We have already proved in Theorem~\ref{th2} that there exists a significantly small positive root for the COPRA function at $\rho_{\text{o,1}} =\epsilon$ but we are not interested in this root. In other words, we would like to see if there exists a second root for $G\left(\rho_{\text{o}}\right)$ in the interval $\left(\epsilon, +\infty\right)$. 

From Property~\ref{p1} and Theorem~\ref{th2}, we can conclude that the COPRA function has a positive value before $\epsilon$, then it switches to the negative region after that. The condition in (\ref{eq:solutioncondition}) guarantees that $G\left(\rho_{\text{o}}\right)$ approaches zero from a positive direction as $\rho_{\text{o}}$ approaches $+\infty$. This means that $G\left(\rho_{\text{o}}\right)$ has an extremum in the interval $\left(\epsilon, +\infty\right)$, and this extremum is actually a minimum point. If the point of the extremum is considered to be $\rho_{\text{o,m}}$, then the function starts increasing for $\rho_{\text{o}} > \rho_{\text{o,m}}$ until it approaches the second zero crossing at $\rho_{\text{o},2}$. Since Theorem~\ref{th1} states clearly that we cannot have more than two roots, we conclude that when (\ref{eq:solutioncondition}) holds, we have only one unique positive root over the interval $\left(\epsilon, +\infty\right)$. 
\end{proof}
\vspace{-10pt}
\subsection{Finding the Root of $G\left(\rho_{\text{o}}\right)$}
\label{subsec:find root}
To find the positive root of the COPRA function $G\left(\rho_{\text{o}}\right)$ in (\ref{eq:IBPRfunction}), Newton's method \cite{zarowski2004introduction} can be used. The function $G\left(\rho_{\text{o}}\right)$ is differentiable in the interval $\left(\epsilon, +\infty\right)$ and the expression of the first derivative $G^{'}\left(\rho_{\text{o}}\right)$ can be easily obtained. Newton's method can then be applied in a straightforward manner to find this root. Starting from an initial value $\rho_{\text{o}}^{n=0} > \epsilon $ that is sufficiency small, the following iterations are performed:
 \begin{equation}
\label{eq:Newton}
\rho_{\text{o}}^{n+1} = \rho_{\text{o}}^{n} - \frac{G\left(\rho_{\text{o}}^{n}\right)}{G^{'}\left(\rho_{\text{o}}^{n}\right)}.
\end{equation}
The iterations stop when $|G\left(\rho_{\text{o}}^{n+1}\right)|< \bar{\xi}$, where $\bar{\xi}$ is a sufficiently small positive quantity.
\vspace{-10pt}
\subsection{Convergence}
\label{subsec:converge}
When condition  (\ref{eq:solutioncondition}) is satisfied, the convergence of Newton's method can be easily proved. As a result from Theorem~\ref{th3}, the function $G\left(\rho_{\text{o}}\right)$ has always a negative value in the interval $\left(\epsilon, \rho_{o,2}\right)$. It is also clear that $G\left(\rho_{\text{o}}\right)$ is an increasing function in the interval $[\rho_{\text{o}}^{n= 0}, \rho_{\text{o},2}]$. Thus, starting from $\rho_{\text{o}}^{n= 0}$, (\ref{eq:Newton}) will produce a consecutively increasing estimate for $\rho_{\text{o}}$. Convergence occurs when $G\left(\rho_{\text{o}}^{n}\right) \rightarrow 0$ and $\rho_{\text{o}}^{n+1}\rightarrow \rho_{\text{o}}^{n}$. When the condition in (\ref{eq:solutioncondition}) is not satisfied, the regularization parameter should be set to $\epsilon$. 
\vspace{-10pt}
\subsection{COPRA Summary}
\label{subsec2:I-BPR summery}
The proposed COPRA discussed in the previous sections is summarized in Algorithm~\ref{COPRA ALGORITHM}.
\begin{algorithm}[t!]
\caption{COPRA summary}
\begin{algorithmic}[t!]
\IF {(\ref{eq:solutioncondition}) is not satisfied}
\STATE $\rho_{\text{o}} = \epsilon$.
\ELSE
\STATE Define $\tilde{\xi}$ as the iterations stopping criterion.
\STATE Set $\rho_{\text{o}}^{n =0}$ to a sufficiently small positive quantity.
\STATE Find {$G\left(\rho_{\text{o}}^{n=0}\right)$} using (\ref{eq:IBPRfunction}), and compute its derivative {$G^{'}\left(\rho_{\text{o}}^{n=0}\right)$.}
\WHILE{$|G\left(\rho_{\text{o}}^{n}\right)| > \tilde{\xi}$}
\STATE Solve (\ref{eq:Newton})  to get $\rho_{\text{o}}^{n+1}$.
\STATE $\rho_{\text{o}}^{n} =\rho_{\text{o}}^{n+1}$.
\ENDWHILE
\ENDIF
\STATE Find $\hat{\xv}$ using (\ref{eq:copra solution}).
\end{algorithmic}
\label{COPRA ALGORITHM}
\end{algorithm}
%\vspace{-8pt}
\section{Numerical Results}
\label{sec:Results}
In this section, we perform a comprehensive set of simulations to examine the performance of the proposed COPRA and compare it with benchmark regularization methods.

Three different scenarios of simulation experiments are performed. Firstly, the proposed COPRA is applied to a set of nine real-world discrete ill-posed problems that are commonly used in testing the performance of regularization methods in discrete ill-posed problems. Secondly, COPRA is used to estimate the signal when $\Am$ is a random rank deficient matrix generated as
\begin{equation}
\Am  = \frac{1}{n}\Bm \Bm^{T},
\end{equation}
where $\Bm \left(m \times n, m > n\right)$ is a random matrix with i.i.d. zero mean unit variance Gaussian random entries. This is a theoretical test example which is meant to illustrate the robustness of COPRA and to make sure that the obtained results are applicable for broad class of matrices. Finally, an image restoration in image tomography ill-posed problem is considered\footnote{The MATLAB code of the COPRA is provided at http://faculty.kfupm.edu.sa/ee/naffouri/publications.html.}.
\vspace{-10pt}
\subsection{Real-World Discrete Ill-posed Problems}
The Matlab regularization toolbox \cite{hansen1994regularization} is used to generate pairs of a matrix $\Am \in \mathbb{R}^{50 \times 50}$ and a signal $\xv_{0}$. The toolbox provides many real-world discrete ill-posed problems that can be used to test the performance of regularization methods. The problems are derived from discretization of Fredholm integral equation as in (\ref{eq:kernal equation}) and they arise in many signal processing applications\footnote{For more details about the test problems consult \cite{hansen1994regularization}.}. 
\begin{figure*}[h!]
	  \centering	 	
	\subfigure[Tomo  (CN = 3.2 $\times 10^{16}$).]{\label{fig:tomonmse}\includegraphics[width=2.3in]{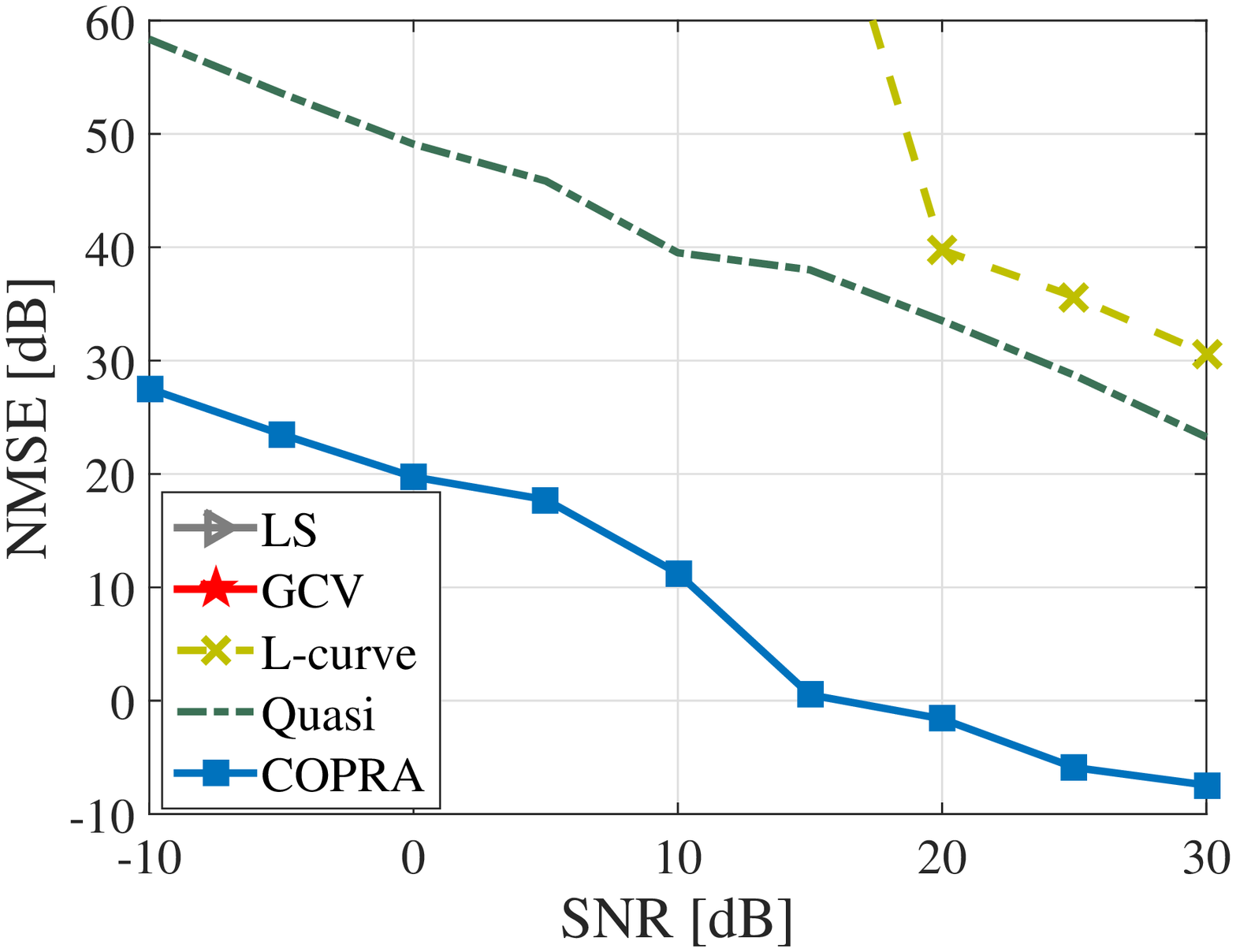}}
	\subfigure[Wing \ (CN = 1.6 $\times 10^{18}$).]{\label{fig:out3}\includegraphics[width=2.3in]{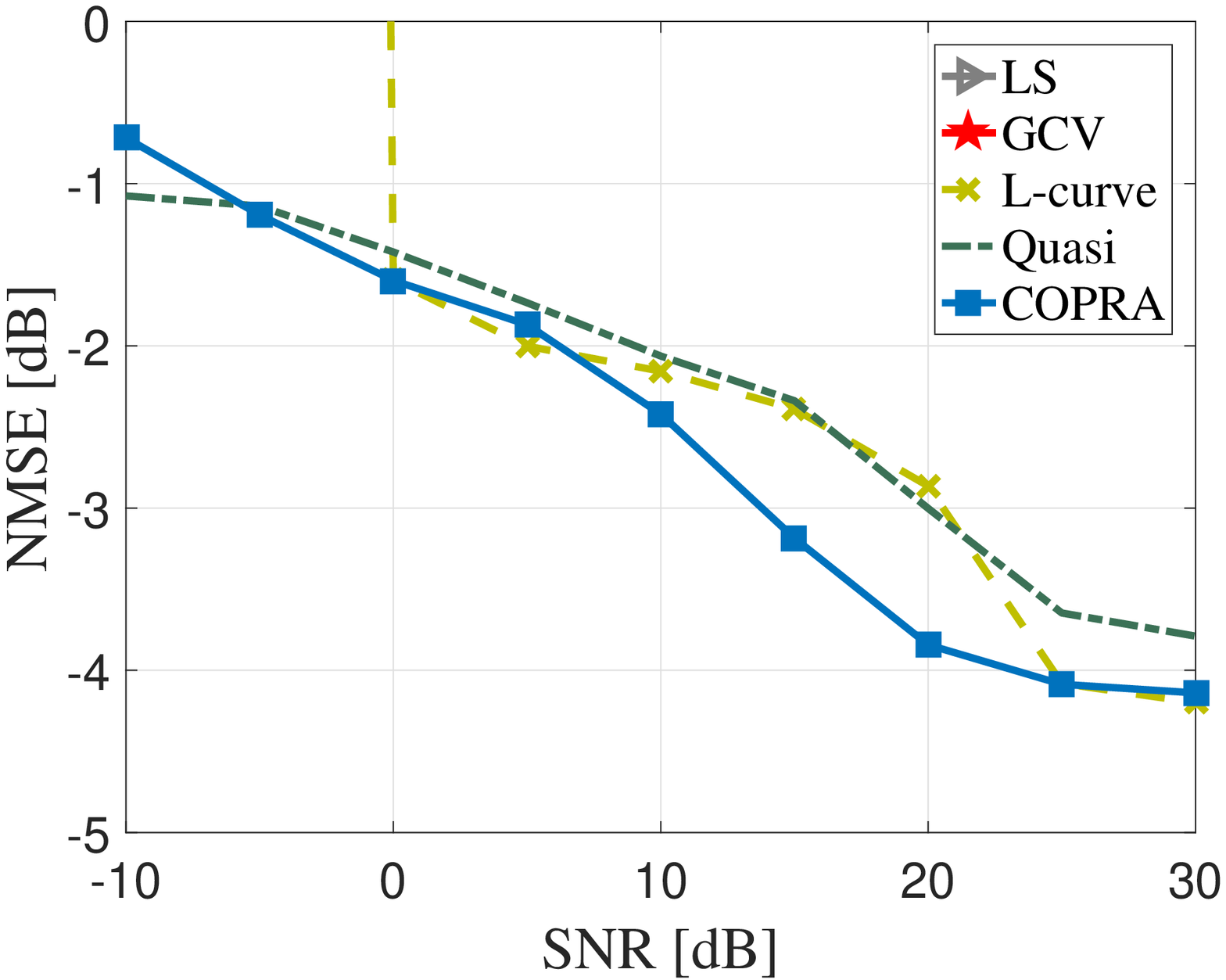}}
	\subfigure[Heat \ (CN = 2.4 $\times 10^{26}$).]{\label{fig:out3}\includegraphics[width=2.3in]{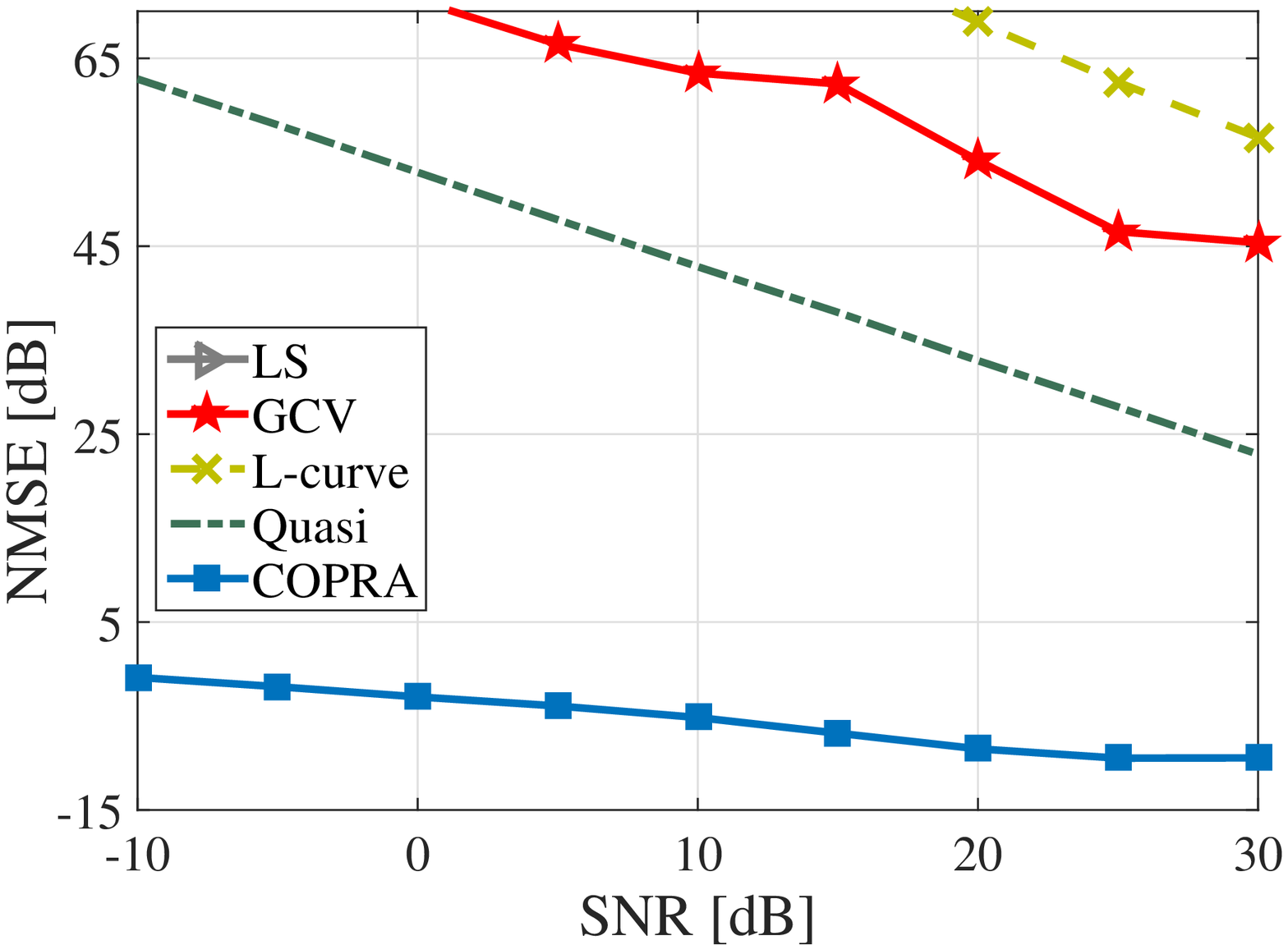}}
	
    \subfigure[Spikes \ (CN = 4.6 $\times 10^{18}$).]{\label{fig:out3}\includegraphics[width=2.3in]{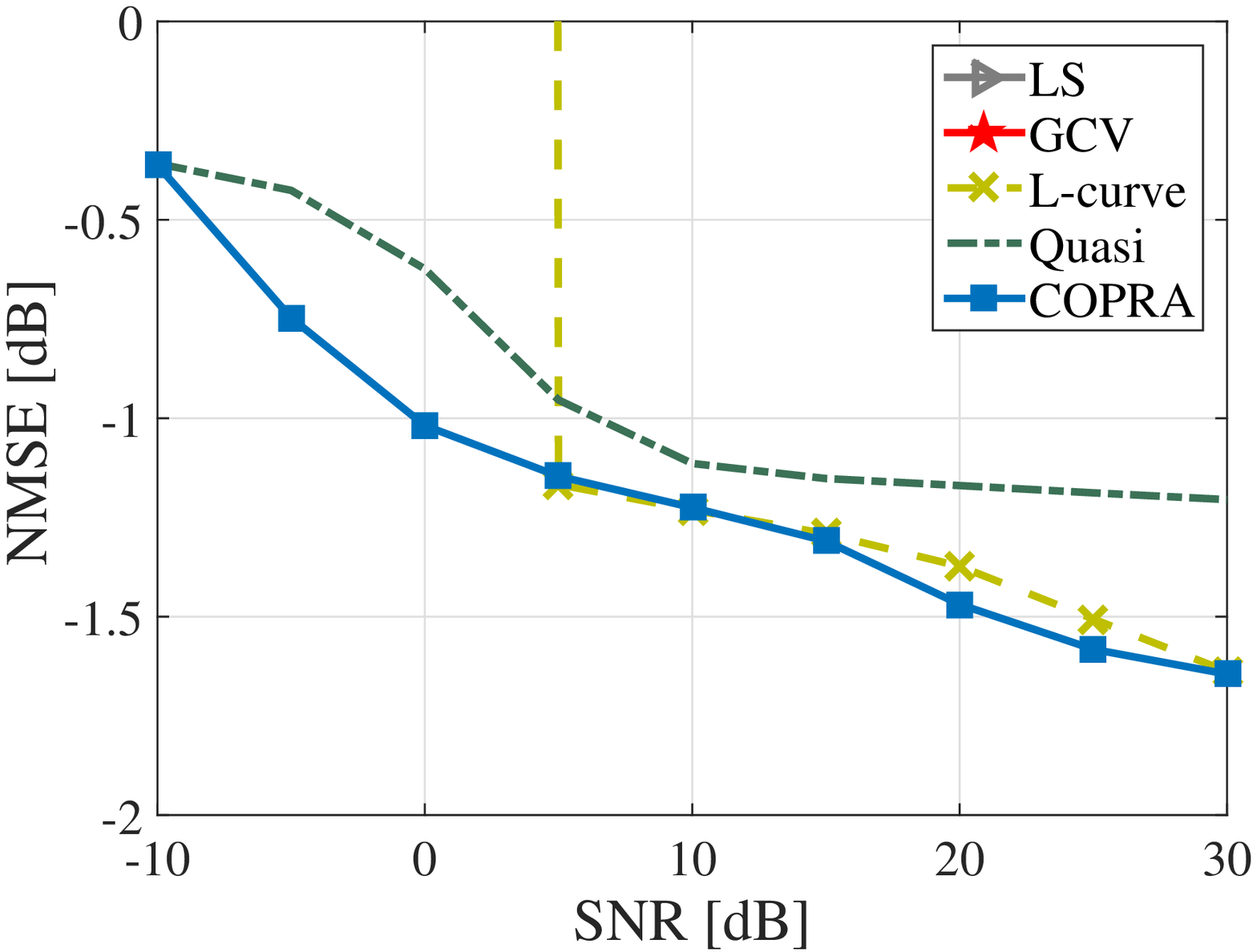}}	
    \subfigure[Baart (CN = 4 $\times 10^{17}$).]{\label{fig:out3}\includegraphics[width=2.3in]{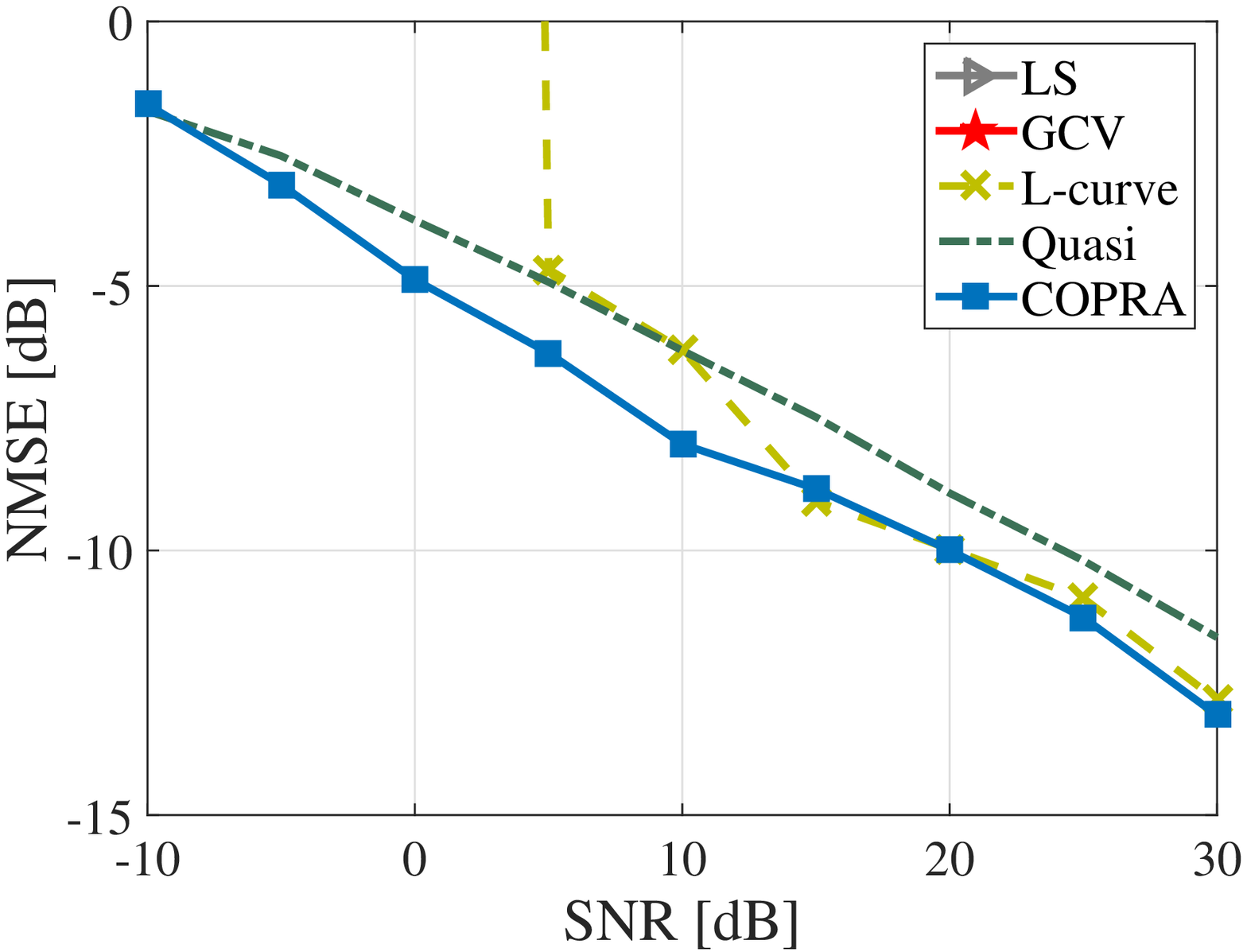}}
    \subfigure[Foxgood \ (CN = 2.4 $\times 10^{18}$).]{\label{fig:out3}\includegraphics[width=2.3in]{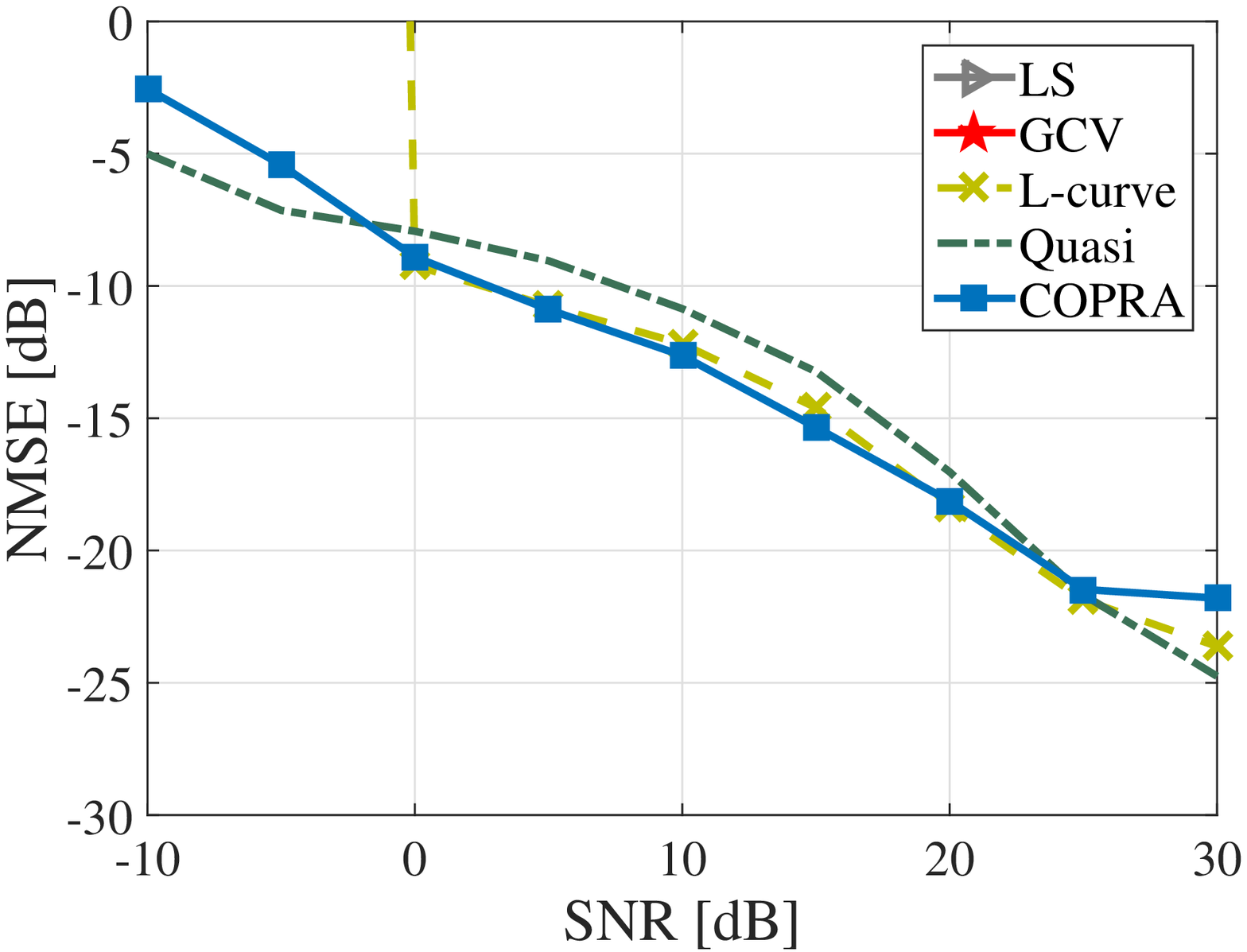}}

	\subfigure[I-Laplace \ (CN = 3.4 $\times 10^{33}$).]{\label{fig:out1}\includegraphics[width=2.3in]{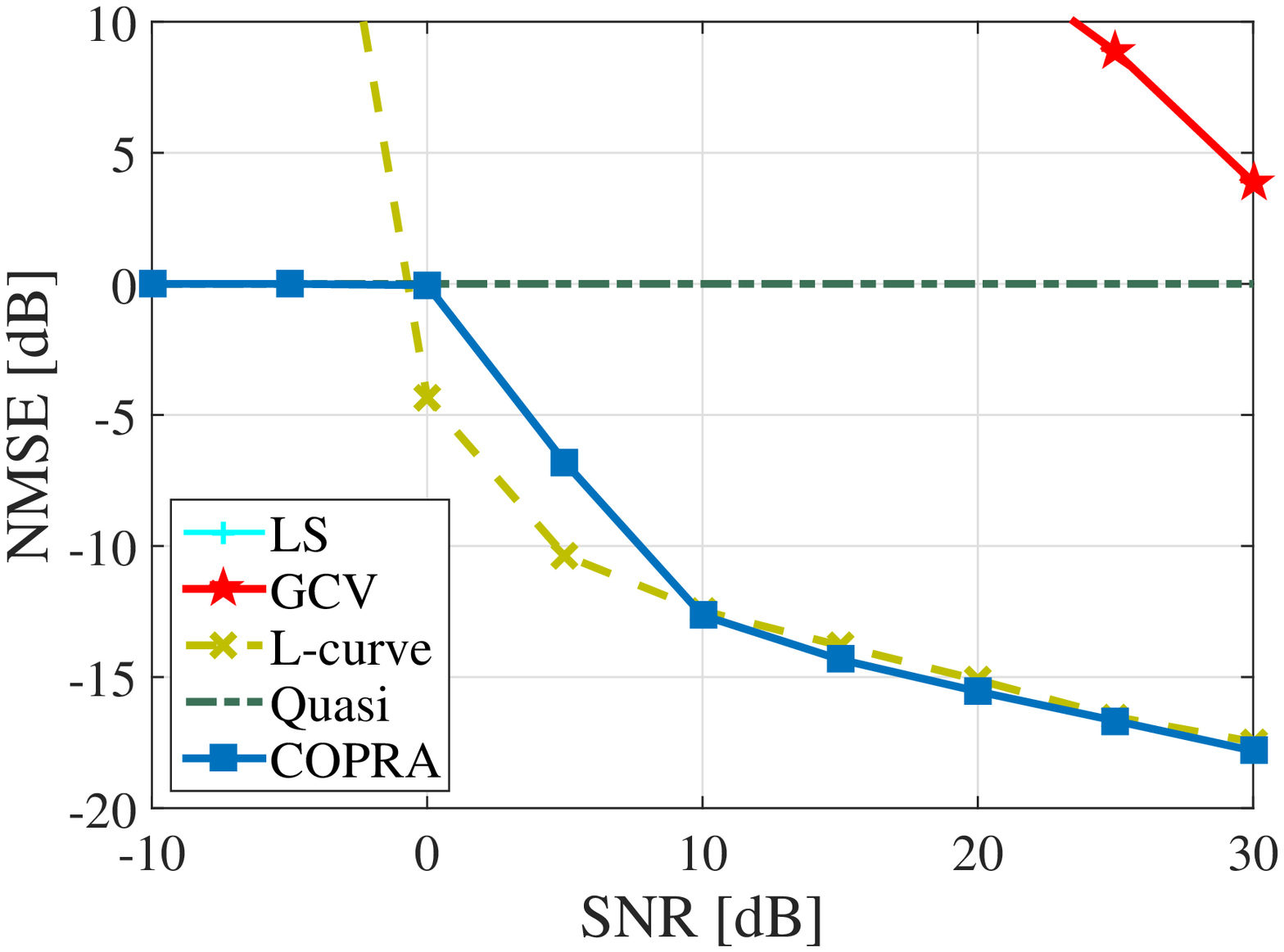}}
	\subfigure[Deriv2 \ (CN = 3 $\times 10^{3}$).]{\label{fig:out2}\includegraphics[width=2.3in]{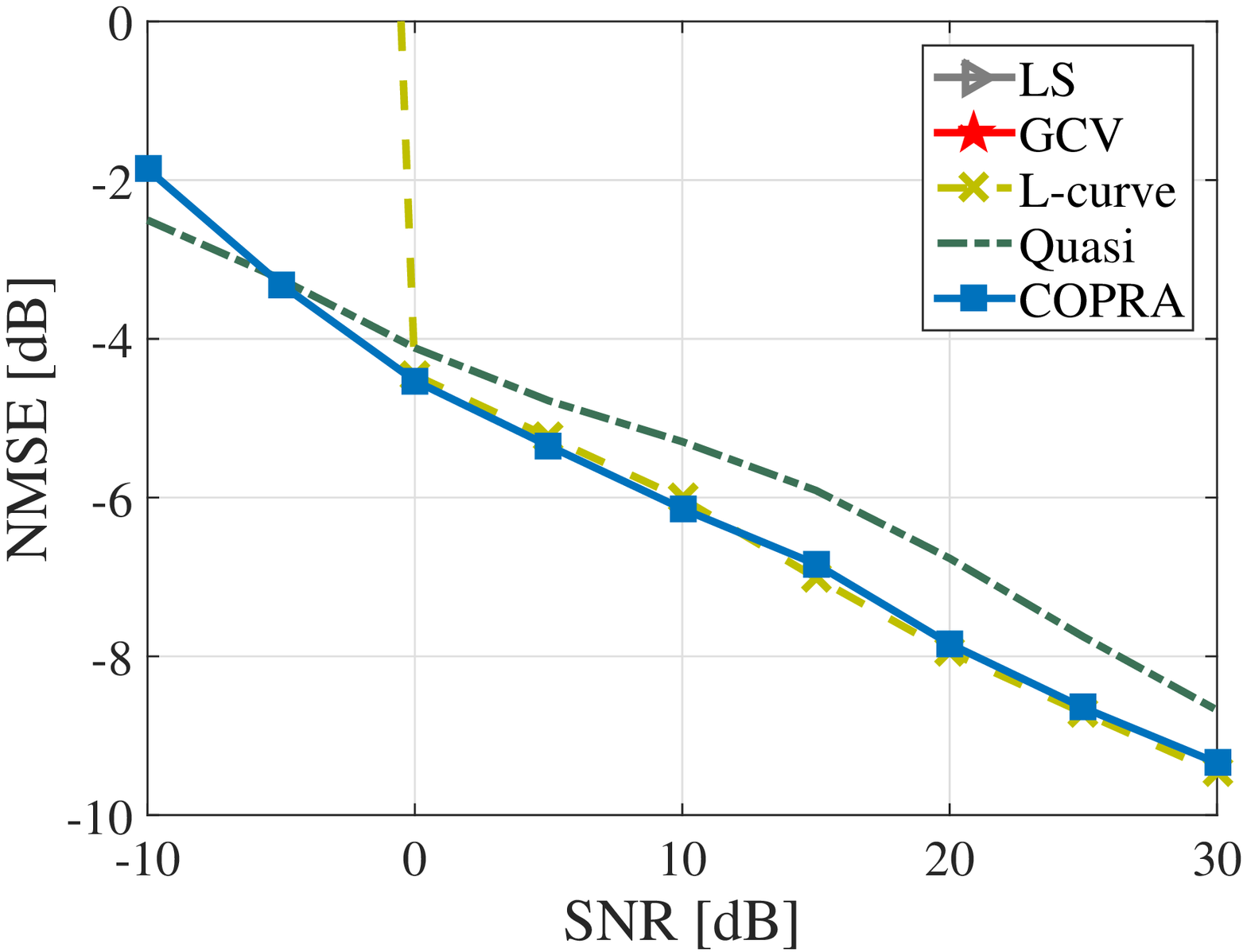}}  
	\subfigure[Shaw \ (CN = 2 $\times 10^{18}$).]{\label{fig:out3}\includegraphics[width=2.3in]{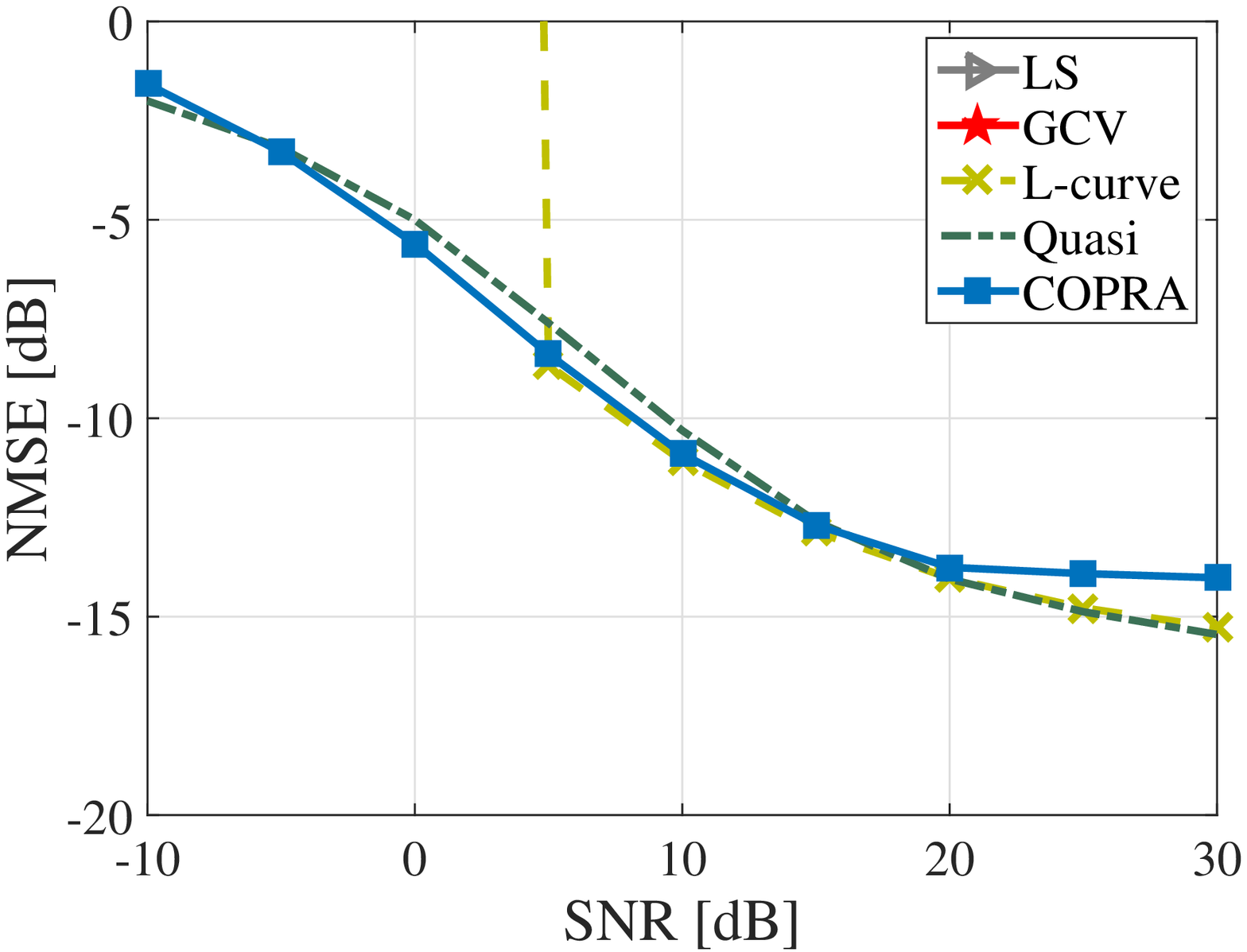}}				
	\caption{ Normalized mean-squared error (NMSE) [dB] vs SNR [dB] (CN $\equiv$ condition number). }
	\label{fig:ill-posed performace}
%	\hrulefill
\end{figure*}

\emph{Experiment setup}: The performance of COPRA is compared with three benchmark regularization methods, the quasi-optimal, the GCV, the L-curve in addition to the LS. The performance is evaluated in terms of normalized MSE (NMSE); that is the MSE normalized by $\lVert \xv_{0} \lVert_{2}^{2}$. Noise is added to the vector $\Am \xv_{0}$ according to a certain signal-to-noise-ratio (SNR) defined as SNR $\triangleq \lVert \Am\xv_{0}\rVert_2^{2}/n \sigma_{\zv}^{2}$ to generate $\yv$. The performance is presented as the NMSE (in dB) (NMSE in dB  = $10\log_{10}$ $\left( \text{NMSE}\right)$) versus SNR (in dB) and is evaluated over $10^{5}$ different noise realizations at each SNR value. Since some regularization methods provide a high unreliable NMSE results that hinder the good visualization of the NMSE, we set different \emph{upper thresholds} for the vertical axis in the results sub-figures.

Fig.~\ref{fig:ill-posed performace} shows the results for all the selected 9 problems. Each sub-figure quote the condition number (CN) of the problem's matrix. The NMSE curve for some methods disappears in certain cases. This indicates extremely poor performance for these methods such that they are out of scale. For example, LS does not show up in all the tested scenarios, while the other benchmark methods disappear in quite a few cases. 

Generally speaking, It can be said that an estimator offering NMSE above 0~dB is not robust and is worthless. From Fig.~\ref{fig:ill-posed performace}, it is clear that COPRA offers the highest level of robustness among all the methods as it is the only approach whose NMSE performance remains below 0~dB in almost all cases. Comparing the NMSE over the range of the SNR values in each problem, we find \emph{on average} that COPRA exhibits the lowest NMSE amongst all the methods in 8 problems (the first 8 sub-figures). Considering all the problems, the closest contender to COPRA is the quasi method. However, this method and the remaining methods show lack of robustness in certain situations as evident by the extremely high NMSE. 
\normalsize
\begin{figure}[h!]
    \centering
  \centerline{\includegraphics[width= 3.8in]{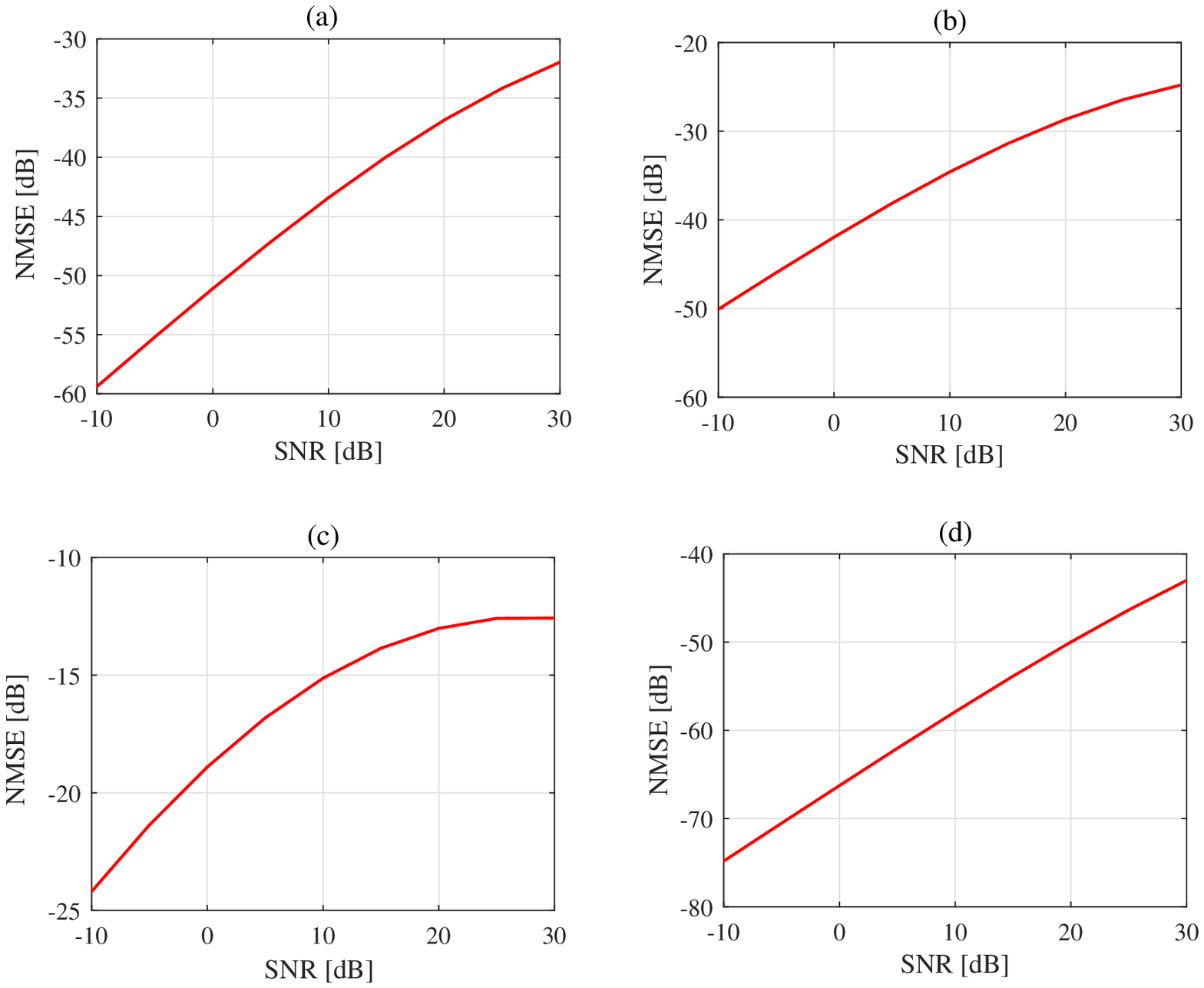}}
\caption{ NMSE [dB] vs SNR [dB] between (\ref{eq:eta min 2}) and (\ref{eq:eta min 3}) for a various $\Am$. (a) Wing problem. (b) Heat problem. (c) Foxgood problem. (d) Deriv2 problem. }
\label{fig:bounds}
\end{figure}

In Fig.~\ref{fig:bounds}, we provide the NMSE for the approximation of the perturbation bound expression in (\ref{eq:eta min 2}) by (\ref{eq:eta min 3}) for a selected ill-posed problem matrices. The two expressions are evaluated at each SNR using the suboptimal regularizer in (\ref{eq:gamma min approx}). The sub-figures show that the NMSE of the approximation is extremely small (below -20 dB in most cases) and that the error increases as the SNR increases. The increase of the approximation error with the SNR is discussed in Appendix~\ref{Apen error}.
\vspace{-9pt}
\subsection{Rank Deficient Matrices}
In this scenario, a rank deficient random matrix $\Am$ is considered. This is the case where $\lVert \Sigmam_{n_{2}} \rVert = 0$ in (\ref{eq:sigma parti}). This theoretical test is meant to illustrate the robustness of COPRA. 

\emph{Experiment setup}: The matrix $\Am$ is generated as a random matrix that satisfies $\Am = \frac{1}{50}\Bm \Bm^{T}$, where $ \Bm  \in \mathbb{R}^{50\times45}, B_{ij} \sim \mathcal{N}\left(0, 1\right)$. The elements of $\xv_{0}$ are chosen to be Gaussian i.i.d. with zero mean unit variance, and i.i.d. with uniform distribution in the interval $(0, 1)$. Results are obtained as an average over $10^{5}$ different realizations of $\Am$, $\xv_{0}$, and $\zv$.

From Fig.~\ref{fig:rank deficinet}(a), we observe that when the elements of $\xv_{0}$ are Gaussian i.i.d. COPRA outperforms all the benchmark regularization methods. In fact, COPRA is the only approach that provides a NMSE below 0 dB overall the SNR range while other algorithms are providing a very high NMSE. The same behavior can be observed when the elements of $\xv_{0}$ are uniformly distributed as Fig.~\ref{fig:rank deficinet}(b) shows. Finally, the performance of the LS is above 250~dB for both cases.
\begin{figure}
	 \centering	 	
	\subfigure[$\xv_{0} \sim \mathcal{N}(\bm{0}, \Id)$ with i.i.d. elements.]{\label{fig:out3}\includegraphics[width=2.2in]{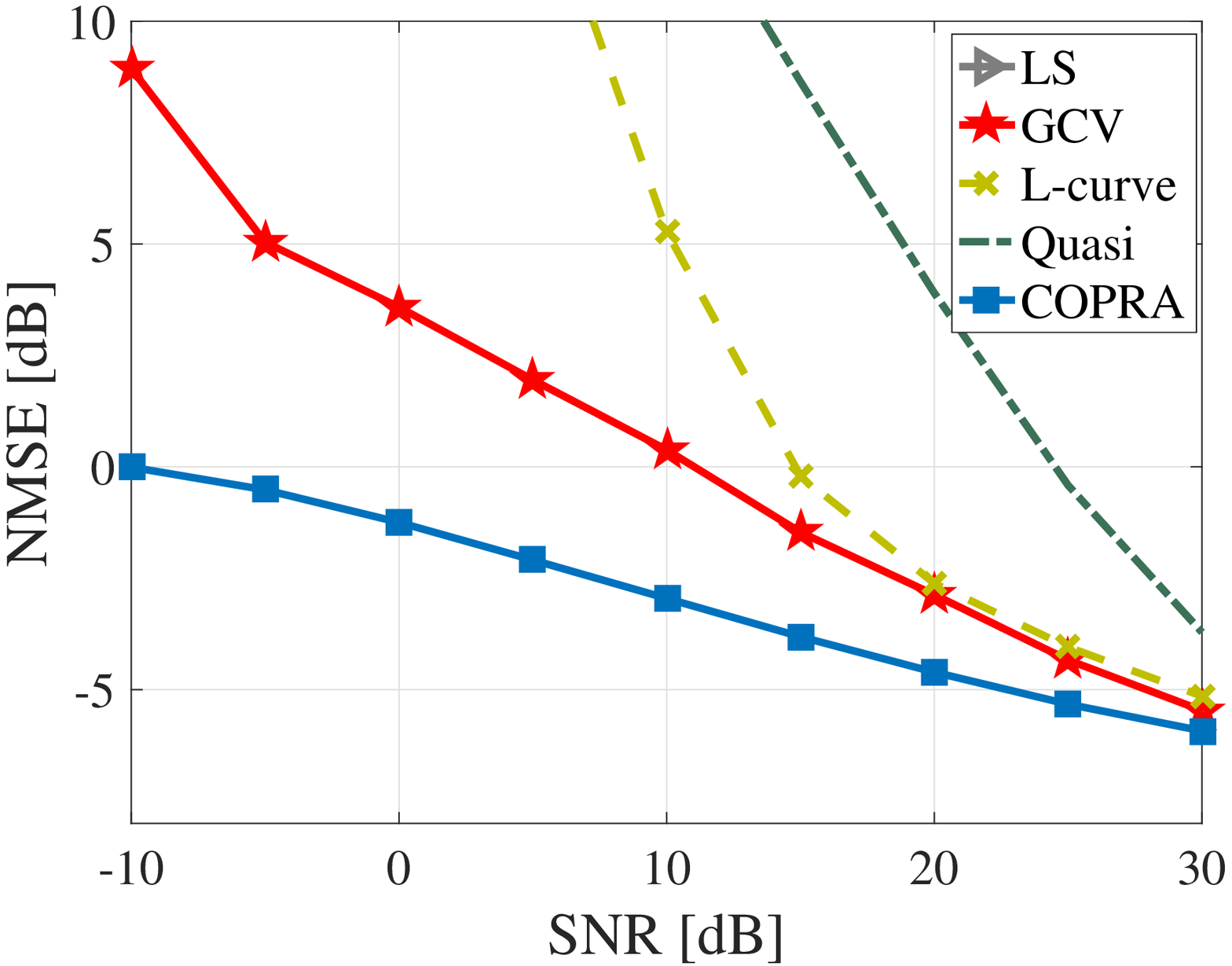}}
	\subfigure[The elements of $\xv_{0}$ are i.i.d. with uniform distribution in the interval $(0, 1)$.]{\label{fig:out3}\includegraphics[width=2.2in]{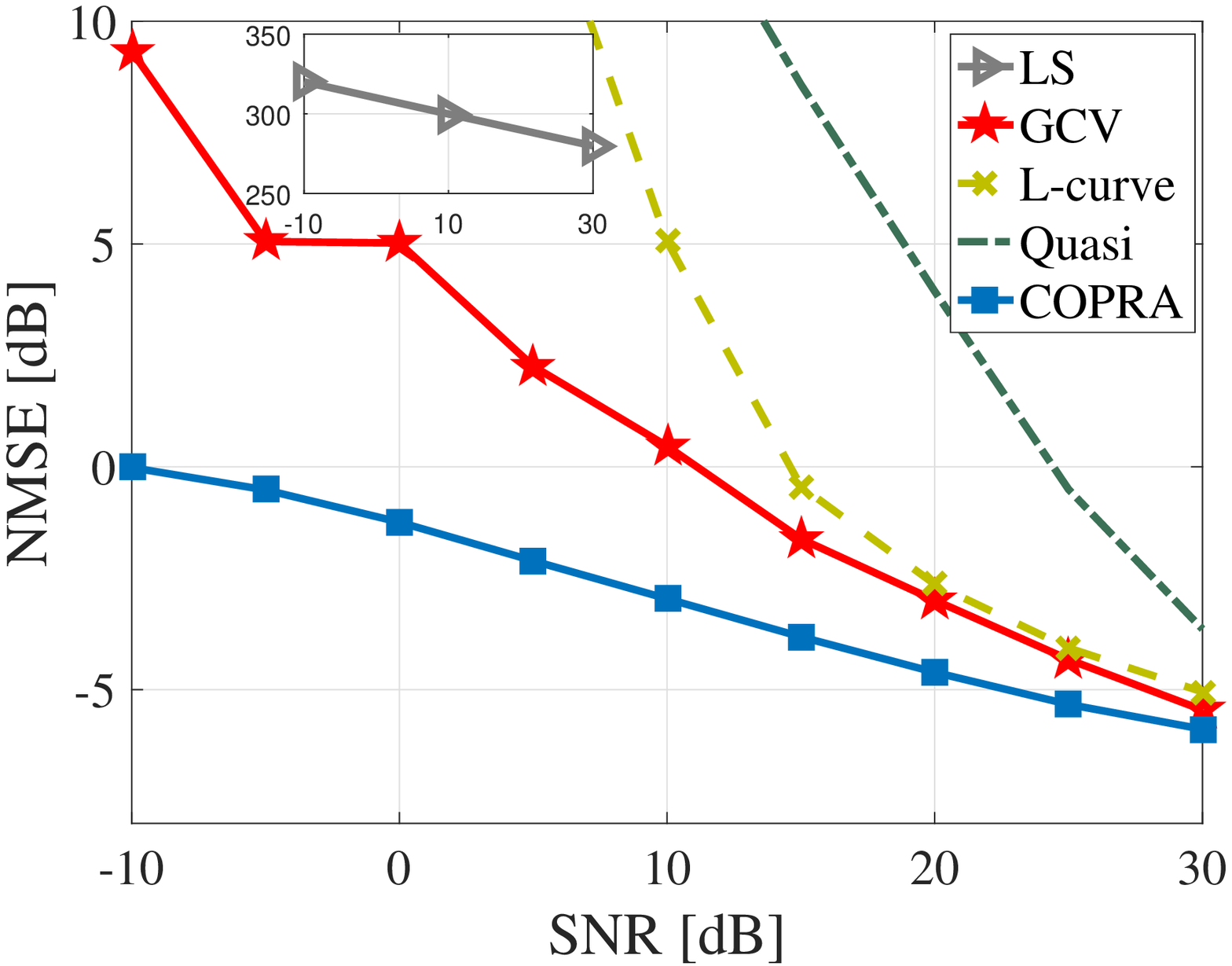}}
			
	\caption{Performance comparison when $\Am = \frac{1}{50}\Bm\Bm^{T}$, where $ \Bm  \in \mathbb{R}^{50\times45}, B_{ij} \sim \mathcal{N}(0, 1)$. }
	\label{fig:rank deficinet}
\end{figure}
\vspace{-9pt}
\subsection{Image Restoration}
The tomo example in Section~\ref{sec:Results}-A and Fig.~\ref{fig:tomonmse} discusses the NMSE of the tomography inverses problem solution. In this subsection, we present visual results for the restored images.

\emph{Experiment setup}: The elements of $\Am \xv_{0} $ are a representative of a line integrals along direct rays that penetrate a rectangular field. This field is discretized into $n^{2}$ cells, and each cell with its own intensity is stored as an element in the image matrix $\Mm$. Then, the columns of $\Mm$ are stacked into $\xv_{0}$. On the other hand, the entries of $\Am$ are generated as 
\begin{equation}
a_{ij} = 
\begin{cases}
&l_{ij},  \  \  \text{pixel}_{j} \in \text{ray}_{i} \\
&0       \  \hspace{10pt} \text{else},
\end{cases}  \nonumber
\end{equation}
where $l_{ij}$ is the length of the $i$'th ray in pixel $j$. Finally, the rays are placed randomly. A noise with SNR equal to 30 dB is added to the image of size $16 \times 16$ and the performance is evaluated as an average over $10^{6}$ noise and $\Am$ realizations.

In Fig.~\ref{fig:image tomo results}, we present the original image, the received image, and the performance of the methods. Fig.~\ref{fig:image tomo results} demonstrates that COPRA outperforms all methods through providing a clear image that is very close to the original image. This also appear when we compare the peak signal-to-noise ratio (PSNR) of the algorithms as in Table.~\ref{tab:psnr} with COPRA having the largest PSNR among them. Moreover, algorithm such GCV provides an unreliable result, while L-curve and quasi fail to restore the internal parts clearly, especially those who have colors close to the image background color.
\begin{figure}
	  \centering	 	
	\subfigure[Original image.]{\label{fig:tomonmse1}\includegraphics[width=.95in]{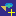}}
	\subfigure[Recieved image.]{\label{fig:tomonmse2}\includegraphics[width=.95in]{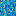}}
	\subfigure[COPRA image.]{\label{fig:tomonmse3}\includegraphics[width=.95in]{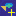}}
	\subfigure[L-curve image.]{\label{fig:tomonmse4}\includegraphics[width=.95in]{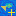}}
	\subfigure[GCV image.]{\label{fig:tomonmse5}\includegraphics[width=.95in]{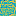}}
	\subfigure[Quasi image.]{\label{fig:tomonmse6}\includegraphics[width=.95in]{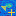}}
	\caption{Image tomography restoration. }
	\label{fig:image tomo results}
\end{figure}

\begin{center}
\captionof{table}{Algorithms PSNR at SNR=30 }
\begin{tabular}[b]{| l || l |l| l | l |}
  \hline
  \rowcolor{lightgray}
    \hline
Method   & COPRA  & L-curve  & GCV   & Quasi  \\\hline
PSNR      & \textbf{29.8331}  & 13.6469  & 10.5410  & 15.9080 \\\hline
  \end{tabular} 
\label{tab:psnr} 
\end{center}
\vspace{-10pt}
\subsection{Average Runtime}
In Fig.~\ref{fig:run time}, we plot the average runtime for each method against the SNR as calculated in the simulation. The figure is a good representation for the runtime of all the problems (no significant runtime variation between problems has been seen). The figure shows that COPRA is the fastest algorithm as it has the lowest runtime in compare to all benchmark methods. 
\vspace{-7pt}
\begin{figure}[h]
    \centering
  \centerline{\includegraphics[width= 2.3in]{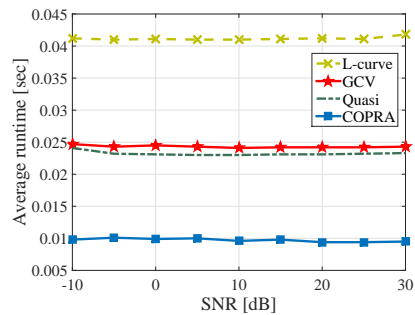}}
\caption{Average runtime.}
\label{fig:run time}
\end{figure}
%\vspace{-20pt}
\section{Conclusion}
\label{sec:conclusion}
In this work, we developed a new approach to find the regularization parameter for linear discrete ill-posed problems. Due to the challenging singular-value structure for such problems, many regularization approaches fail to provide a good stabilize solution. In the proposed approach, the singular-value structure of the model matrix is modified by allowing an artificial perturbation into it. To maintain the fidelity of the model, an upper bound constraint on the perturbation is allowed. The proposed approach minimizes the worst-case residual error of the estimator and selects the perturbation bound in a way that approximately minimizes the MSE. As a result, the approach combines the simplicity of the least-squares criterion with the robustness of the MSE based estimation. The regularization parameter is obtained as a solution of a non-linear equation in one unknown variable. Simulation results demonstrate that the proposed approach outperforms a set of benchmark regularization methods.
%We study the properties of this equation and show that Newton's method converges to the desired root. 
\vspace{-15pt}
\appendices
\numberwithin{equation}{section}
\section{Error Analysis}
\label{Apen error}

In this appendix, we analyze the error of the approximation that is made to obtain (\ref{eq:gamma min approx}). To simpify the analysis, we will consider the case when the approximation is applied directly to the MSE function in (\ref{eq:MSE2}). Let us start by defining $\Hm\triangleq {\Sigmam}^{2k} \left({\Sigmam}^2 + \rho \Id_n \right)^{-p}$, whose diagonal entries can be written as
\begin{equation}
\label{eq:diagonal}
h_{ii} = \frac{\sigma_i^{2k}}{\left(\sigma_i^2+\rho\right)^p} ; \  i=1,2,\cdots, n.
\end{equation}
Note that in our case $k=0$ and $p=2$ for the diagonal matrix inside the trace function of the second term in (\ref{eq:MSE2}). However, we will use these two variables to obtain the error expression for the general case, then we will substitute for $k$ and $p$. By using the inequalities in [\cite{wang1986trace}, eq.(5)], we obtain
\begin{equation}
\label{eq:eneq 1}
\lambda_\text{min}(\Rm_{\xv_{0}}) \text{Tr}\left( \Hm \right)
 \leq
 \text{Tr}\left( \Hm \Vm^T\Rm_{\xv_{0}}\Vm \right)
\leq
\lambda_\text{max}(\Rm_{\xv_{0}}) \text{Tr}\left(\Hm \right).
\end{equation}
Similarly, we can write
\begin{equation}
\label{eq:eneq 2}
\lambda_\text{min}\left(\Hm\right) \text{Tr}\left( \Rm_{\xv_{0}} \right)
 \leq
 \text{Tr}\left( \Hm \Vm^T\Rm_{\xv_{0}}\Vm \right)
\leq
\lambda_\text{max}\left(\Hm\right) \text{Tr}\left(\Rm_{\xv_{0}} \right).
\end{equation}
Since $\Hm$ is diagonal, $\lambda_\text{min}\left(\Hm\right) = \text{min}\left(\text{diag}(\Hm)\right)$ and $\lambda_\text{max}\left(\Hm\right) = \text{max}\left(\text{diag}(\Hm)\right)$. Now, let us define the normalized error of the approximation as
\begin{equation}
\label{eq:error def}
\varepsilon = \frac{\text{Tr}\left( \Hm \Vm^T\Rm_{\xv_{0}}\Vm \right) - \frac{1}{n}\text{Tr}\left( \Hm \right) \text{Tr}\left( \Rm_{\xv_{0}} \right) }
 {\frac{1}{n}\text{Tr}\left( \Hm \right) \text{Tr}\left( \Rm_{\xv_{0}} \right)}.
\end{equation}
Note that this is not the standard way of defining the normalized error. Typically, the error $\varepsilon$ is normalized by the true quantity, i.e., $\text{Tr}\left( \Hm \Vm^T\Rm_{\xv_{0}}\Vm \right)$. However, this way of defining the error is found to be more useful in carrying out the following error analysis. Based on (\ref{eq:error def}), we see that $|\varepsilon| \geq 1$ indicates an inaccurate approximation. Although at the end it depends totally on the application, we will adopt $|\varepsilon| < 1$ as the reference for evaluating the accuracy of the approximation. In fact, we can observe from (\ref{eq:error def}) that $|\varepsilon| = 1$ indicates that $\frac{1}{n}\text{Tr}\left( \Hm \right) \text{Tr}\left( \Rm_{\xv_{0}} \right) = 0.5  \ \text{Tr}\left( \Hm \Vm^T\Rm_{\xv_{0}}\Vm \right) $. To this end, we will derive two bounds based on \eqref{eq:eneq 1} and \eqref{eq:eneq 2}. Then, we will combine them to obtain the final error bound.

\emph{Absolute Error Bound Based on \eqref{eq:eneq 1}}:
Subtracting $\frac{1}{n}\text{Tr}\left( \Hm \right) \text{Tr}\left( \Rm_{\xv_{0}} \right)$ from \eqref{eq:eneq 1} and dividing by the same quantity, we obtain
\begin{equation}
\label{eq:eneq 1 bound}
\frac{\lambda_\text{min}(\Rm_{\xv_{0}})}{\lambda_\text{avg}(\Rm_{\xv_{0}})}-1
 \leq
\varepsilon
\leq
\frac{\lambda_\text{max}(\Rm_{\xv_{0}})}{\lambda_\text{avg}(\Rm_{\xv_{0}})}-1,
\end{equation}
where $\lambda_\text{avg}(\Rm_{\xv_{0}}) \triangleq \frac{1}{n} \text{Tr}\left( \Rm_{\xv_{0}}\right)$. Thus, $|\varepsilon|$ can be bounded by a positive quantity according to
\begin{equation}
\label{eq:eneq 1 abs bound}
|\varepsilon_{x}|
 \leq \mu_x = \max
 \left[ 1 - \frac{ \lambda_\text{min}(\Rm_{\xv_{0}}) }{\lambda_\text{avg}(\Rm_{\xv_{0}})},
 \frac{ \lambda_\text{max}(\Rm_{\xv_{0}}) }{\lambda_\text{avg}(\Rm_{\xv_{0}})} -1
  \right].
\end{equation}

\emph{Absolute Error Bound Based on \eqref{eq:eneq 2}}:
Starting from \eqref{eq:eneq 2}, and by applying the same approach used to obtain \eqref{eq:eneq 1 abs bound}, we derive the second bound as
 \begin{equation}
\label{eq:eneq 2 abs bound 0}
|\varepsilon_{a}|
 \leq \mu_a = \max
 \left[ 1 - \frac{ \lambda_\text{min}(\Hm) }{\lambda_\text{avg}(\Hm)},
 \frac{ \lambda_\text{max}(\Hm) }{\lambda_\text{avg}(\Hm)} -1
  \right].
\end{equation}
By using \eqref{eq:diagonal} we can transform \eqref{eq:eneq 2 abs bound 0} to
\begingroup
    \fontsize{9.35pt}{9.35pt} 
 \begin{align}
\label{eq:eneq 2 abs bound}
|\varepsilon_{a}|
 \leq \mu_a & = \max
 \left[1 -
 \frac{ \underset{i}{\min}\left[ \frac{\sigma_i^{2k}}{(\sigma_i^2 +\rho)^p}\right] }
 { \frac{1}{n}\sum_{i=1}^{n} \frac{\sigma_i^{2k}}{(\sigma_i^2 +\rho)^p} },
 \frac{ \underset{i}{\max}\left[ \frac{\sigma_i^{2k}}{(\sigma_i^2 +\rho)^p}\right]}
 { \frac{1}{n}\sum_{i=1}^{n} \frac{\sigma_i^{2k}}{(\sigma_i^2 +\rho)^p} }
  -1 \right]\nonumber\\
  & i=1, 2,\cdots, n.
 \end{align}
 \endgroup
Note that the bound $\mu_x$ depends only on $\Rm_{\xv_{0}}$, while $\mu_a$ depends both the singular values of $\Am$ and the unknown regularizaer $\rho$. As in our case, $k=0$ and $p=2$, and therefore, (\ref{eq:eneq 2 abs bound}) can be simplified to
\begingroup
    \fontsize{9.35pt}{9.35pt} 
 \begin{align}
\label{eq:eneq 2 abs bound k0}
|\varepsilon_{a}|
 \leq \mu_{a} & = \max
 \left[1 -
 \frac{ \frac{1}{(\sigma_1^2 +\rho)^2} }
 { \frac{1}{n}\sum_{i=1}^{n} \frac{1}{(\sigma_i^2 +\rho)^2} },
 \frac{ \frac{1}{(\sigma_n^2 +\rho)^2}}
 { \frac{1}{n}\sum_{i=1}^{n} \frac{1}{(\sigma_i^2 +\rho)^2} }
  -1 \right]
 \end{align}
 \endgroup  

\emph{Combined Bound}:
By combining \eqref{eq:eneq 1 abs bound} and \eqref{eq:eneq 2 abs bound k0}, we obtain the final bound on the absolute error as
 \begin{equation}
\label{eq:eneq combined bound}
|\varepsilon|  \leq \mu =  \min \left( \mu_x, \mu_a \right ).
\end{equation}
From \eqref{eq:eneq 1 abs bound}, \eqref{eq:eneq 2 abs bound k0}, and  \eqref{eq:eneq combined bound}, we notice that the bound is a minimum of two independent bounds. We will analyze each bound separately and then we will derive a conclusion concerning the overall error bound.
\vspace{-10pt}
\subsection{Analysis of $\mu_x$} When $\xv_{0}$ is deterministic, $\lambda_\text{min}(\Rm_{\xv_{0}}) = 0$, $\lambda_\text{max}(\Rm_{\xv_{0}}) = || \xv_{0} ||_{2}^{2}$ and $\lambda_\text{avg}(\Rm_{\xv_{0}}) = \frac{1}{n}|| \xv_{0} ||_{2}^{2}$. By substituting in \eqref{eq:eneq 1 abs bound}, we obtain
\begin{equation}
\label{eq:eneq 1 abs bound deterministic}
\mu_x = \max
 \left[1, n-1\right] = n-1.
\end{equation}
On the other hand, when $\xv_{0}$ is stochastic with i.i.d. elements, $\lambda_\text{min}(\Rm_{\xv_{0}})= \lambda_\text{avg}(\Rm_{\xv_{0}})= \lambda_\text{min}(\Rm_{\xv_{0}}) =\sigma_{\xv_{0}}^{2}$, and as a result 
\begin{equation}
\label{eq:eneq 1 abs bound iid}
\mu_x = \max
 \left[0, 0\right] =0,
\end{equation}
which means based on (\ref{eq:eneq combined bound}) that the approximation is exact regardless of the contribution of the error from $\mu_a$. When $\xv_{0}$ deviate from being i.i.d., it will be very difficult to obtain a value for $\mu_x$. Therefore, and since no previous knowledge about $\xv_{0}$ is assumed in this paper, it seems that this bound is very loose for a general $\xv_{0}$ and we should rather rely on $\mu_a$ to tighten and evaluate the bound of the error as we will discuss. Thus, the modified bound, which can be larger than the actual bound, is given by
\begin{equation}
\label{eq:eneq combined bound modified}
|\varepsilon|  \leq \mu =  \mu_a.
\end{equation}
\vspace{-30pt}
\subsection{Analysis of $\mu_a$}
By taking the derivative of each of the two terms inside (\ref{eq:eneq 2 abs bound k0}) w.r.t. $\rho$, we can easily prove that the two functions are decreasing in $\rho$. This means that we can obtain the two extreme error bounds (the largest and the smallest possible value of the absolute error) by analyzing the two extreme SNR scenarios, i.e., the high SNR regime and the low SNR regime.
\subsubsection{Analysis for the low SNR regime}
In the extreme low SNR regime, we have $\rho\rightarrow \infty$, and therefore, we can obtain the minimum bound on the absolute error. Based on \eqref{eq:diagonal} we can write 
\begin{equation}
\label{eq:D diag gammaInf}
h_{ii}  = \frac{1}{(\sigma_i^2 + \rho)^2}\rightarrow \frac{1}{\rho^2 } ; \  i=1,2,\cdots, n.
\end{equation}
Consequently, \eqref{eq:eneq 2 abs bound k0} will boil down after some manipulations to
\begin{align}
\label{eq:eneq 2 abs bound exp  large gamma 2}
|\varepsilon_{a}^{l}|
 \leq \mu_{a}^{l} & = \max
 \left[1 -
 \frac{  \frac{1}{\rho^{2}} }
 { \frac{1}{n}\sum_{i=1}^{n} \frac{1}{\rho^{2}} },
 \frac{  \frac{1}{\rho^{2}}}
 { \frac{1}{n}\sum_{i=1}^{n} \frac{1}{\rho^{2}} }
  -1 \right]\nonumber\\   
  & = \max
 \left[1 -
 \frac{n}
 { \sum_{i=1}^{n}1 },
 \frac{ n}
 { \sum_{i=1}^{n} 1 }
  -1 \right] = 0.
\end{align}
The result in (\ref{eq:eneq 2 abs bound exp  large gamma 2}) indicates that as the SNR decreases, the approximation becomes more accurate. In the extreme low SNR regime (i.e., $\rho \to \infty$), the absolute error is zero and the approximated term is exactly equal to the original one. 

\subsubsection{Analysis for the high SNR regime}
At the extremely high SNR, $\rho$ is sufficiently small and thus doing such analysis allows us to obtain the upper worst case possible value for $\varepsilon$.  

Since the main objective of this paper is to provide an estimator that minimizes the MSE, and based on the result proven in \cite{hoerl1970ridge}, there always exists a positive regularizer $\rho > 0$, such that the regularized estimator offers lower MSE than the OLS estimator. This implies also for well-conditioned problems. However, if the condition number is too small, both the regularization parameter and the corresponding improvement in the MSE performance comparing to that of the OLS are too small. Therefore, we conclude that for the extreme high SNR, $\rho$ converges to a minimum value $\rho_{\text{min}}$. In what follow, we find a lower bound expression for $\rho_{\text{min}}$ and then examine the absolute error in this value.

Starting from the definition of the SNR, we can write 
\begin{equation}
\label{eq:SNR}
\text{SNR} = \frac{|| \Am \xv_{0} ||_{2}^{2}}{|| \zv ||_{2}^{2}}.
\end{equation}
Applying the SVD of $\Am$ to (\ref{eq:SNR}) then doing some algebraic manipulations, we obtain
\begin{equation}
\label{eq:SNR2}
\text{SNR} = \frac{\text{Tr}\left(\Vm \Sigmam^{2} \Vm^{T} \Rm_{\xv_{0}}\right)}{n  \sigma_{\zv}^{2}},
\end{equation}
where $\Rm_{\xv_{0}}= \xv_{0} \xv_{0}^{T}$. Now, using (\ref{eq:SNR2}) with the suboptimal regularizer $\rho_{\text{o}}$ expression in (\ref{eq:gamma min approx}), we can write
\begin{equation}
\label{eq:regularizer based SNR}
\rho_{\text{o}} = \frac{1}{\text{SNR}} \frac{\text{Tr}\left(\Vm \Sigmam^{2} \Vm^{T} \Rm_{\xv_{0}}\right)}{\text{Tr}\left(\Rm_{\xv_{0}}\right)}.
\end{equation}
From (\ref{eq:regularizer based SNR}) we deduce that for a given $\Am$ and $\xv_{0}$, the minimum achievable suboptimal regularizer $\rho_{\text{min}}$ dependents on the maximum SNR (i.e., $\text{SNR}_{\text{max}}$). That is
\begin{equation}
\label{eq:regularizer based SNR2}
\rho_{\text{min}} = \frac{1}{\text{SNR}_{\text{max}}} \frac{\text{Tr}\left(\Vm \Sigmam^{2} \Vm^{T} \Rm_{\xv_{0}}\right)}{\text{Tr}\left(\Rm_{\xv_{0}}\right)}.
\end{equation}
Give the nature of the ill-posed problems and there singular values behavior as Fig~\ref{fig:sv decay} shows, we will partition $\Sigmam$ and $\Vm$ into to two sub-matrices (same as in Section \ref{choosing optimal bound}) and then approximate (\ref{eq:regularizer based SNR2}) as  
\begin{equation}
\label{eq:regularizer based SNR3}
\rho_{\text{min}} \approx \frac{1}{\text{SNR}_{\text{max}}} \frac{\text{Tr}\left(\Vm_{n_{1}} \Sigmam_{n_{1}}^{2} \Vm_{n_{1}}^{T} \Rm_{\xv_{0}}\right)}{\text{Tr}\left(\Rm_{\xv_{0}}\right)},
\end{equation}
where $\Sigmam_{n_{1}}^{2} =  \text{diag}\left({\sigma}_{1}^{2}, \dots, {\sigma}_{n_{1}}^{2}  \right)$. The value of $\rho_{\text{min}}$ in (\ref{eq:regularizer based SNR3}) can be bounded by 
\begingroup
    \fontsize{9.35pt}{9.35pt} 
\begin{equation}
\label{eq:regularizer based SNR3 bounds}
\frac{{\sigma}_{n_{1}}^{2}}{\text{SNR}_{\text{max}}} \frac{\text{Tr}\left( \Vm_{n_{1}}^{T} \Rm_{\xv_{0}} \Vm_{n_{1}}\right)}{\text{Tr}\left(\Rm_{\xv_{0}}\right)}  \leq  \rho_{\text{min}} \leq \frac{{\sigma}_{1}^{2}}{\text{SNR}_{\text{max}}} \frac{\text{Tr}\left( \Vm_{n_{1}}^{T} \Rm_{\xv_{0}} \Vm_{n_{1}}\right)}{\text{Tr}\left(\Rm_{\xv_{0}}\right)}
\end{equation}
\endgroup
Since we are considering the worst case upper bound for the absolute error, and given that this absolute error increases as $\rho$ decreases, we will consider the lower bound of $\rho_{\text{min}}$ as in (\ref{eq:regularizer based SNR3 bounds}). Moreover, and based on the unitary matrix property and the partitioning of $\Vm$, we can obtain a lower bound for the lower bound in (\ref{eq:regularizer based SNR3 bounds}) as 
\begin{equation}
\label{eq:regularizer based SNR3 bounds 2}
\rho_{\text{min}} \geq  \frac{{\sigma}_{n_{1}}^{2}}{\text{SNR}_{\text{max}}} \frac{\text{Tr}\left( \Vm_{n_{1}}^{T} \Rm_{\xv_{0}} \Vm_{n_{1}}\right)}{\text{Tr}\left(\Rm_{\xv_{0}}\right)} \geq \frac{{\sigma}_{n_{1}}^{2}}{\text{SNR}_{\text{max}}} \frac{\text{Tr}\left(\Rm_{\xv_{0}} \right)}{\text{Tr}\left(\Rm_{\xv_{0}}\right)}.
\end{equation}
Thus, a lower bound for $\rho_{\text{min}}$ (i.e., $\rho_{\text{min}}^{l}$) can be written as
\begin{equation}
\label{eq:regularizer based SNR3 bounds 3}
\rho_{\text{min}}^{l} = \frac{{\sigma}_{n_{1}}^{2}}{\text{SNR}_{\text{max}}}.
\end{equation}
Now, we are ready to study the behavior of the error in the high SNR regime. When $\rho \to \rho_{\text{min}}^{l}$, (\ref{eq:eneq 2 abs bound k0}) can be written as
\begingroup
    \fontsize{8.8pt}{8.8pt} 
 \begin{align}
\label{eq:eneq 2 abs bound exp  small gamma}
|\varepsilon_{a}^{h}|
 \leq  \max
 \left[1 -
 \frac{  \frac{1}{(\sigma_1^2 +\rho_{\text{min}}^{l})^2} }
 { \frac{1}{n}\sum_{i=0}^{n-1} \frac{1}{(\sigma_i^2 +\rho_{\text{min}}^{l})^2} },
 \frac{  \frac{1}{(\sigma_{n}^2 +\rho_{\text{min}}^{l})^2}}
 { \frac{1}{n}\sum_{i=0}^{n-1} \frac{1}{(\sigma_i^2 +\rho_{\text{min}}^{l})^2} }
  -1 \right]
 \end{align}
\endgroup 
To evaluate (\ref{eq:eneq 2 abs bound exp  small gamma}), we will relay on numerical results. Firstly, let us consider $\text{SNR}_{\text{max}}=40$ dB which is a realistic upper value in many signal processing and communication applications. Substituting this value in (\ref{eq:regularizer based SNR3 bounds 3}) we find that $\rho_{\text{min}}^{l}= 0.018 {\sigma}_{n_{1}}^{2}$. Now, substituting the result in (\ref{eq:eneq 2 abs bound exp  small gamma}), then evaluating the expression for the 9 ill-posed problems in Section~\ref{sec:Results}, we find that $|\varepsilon_{a}^{h}| \leq \mu_{a}^{h} \approx 1$. As this represents the worst case upper value for the absolute error, we conclude that
\begin{equation} 
\label{eq:fd}
|\varepsilon_{a}^{h}| \leq \varrho;  \  \   \  \varrho < 1.
\end{equation}
Finally, based on (\ref{eq:eneq combined bound}) and (\ref{eq:eneq combined bound modified}), and by combining (\ref{eq:eneq 2 abs bound exp  large gamma 2}) and (\ref{eq:fd}), we conclude that
\begin{equation}
\label{eq:finall error}
|\varepsilon|\in [0, \varrho];  \   \   \  \varrho <1.
\end{equation}
The conclusion in (\ref{eq:finall error}) indicates that $\frac{1}{n}\text{Tr}\left( \Hm \right) \text{Tr}\left( \Rm_{\xv_{0}} \right) =  q \ \text{Tr}\left( \Hm \Vm^T\Rm_{\xv_{0}}\Vm \right) $ where $q \in (0.5, 1]$.
\vspace{-5pt}
\section{Proof of theorem \ref{th2}}
\label{Apen A}
We are interested in studying the behavior of $\lim_{\rho_{\text{o}} \to  \epsilon} G\left(\rho_{\text{o}}\right)$, assuming that  $ \epsilon $ is sufficiently small positive number (i.e., $ \epsilon \to 0^{+} $), and $ \epsilon \ll \sigma_{i}^{2}$, $\forall i \in [1,n]$. Starting from COPRA function in (\ref{eq:IBPRfunction}), and by defining $\bv = \Um^{T}\yv$, we can write
\begin{align}
\label{eqC:MSE sum 1}
&G\left(\rho_{\text{o}}\right)
=
\sum_{i=1}^{n}\frac{\sigma_{i}^{2} b_{i}^{2}}{\left(\sigma_{i}^{2}+\rho_{\text{o}}\right)^{2}}\sum_{j=1}^{n_{1}}\frac{\left(\beta \sigma_{j}^{2}+\rho_{\text{o}}\right)}{\left(\sigma_{j}^{2}+\rho_{\text{o}}\right)^{2}} \nonumber\\
&-
\sum_{i=1}^{n}\frac{b_{i}^{2}}{\left(\sigma_{i}^{2}+\rho_{\text{o}}\right)^{2}}\sum_{j=1}^{n_{1}}\frac{\sigma_{j}^{2}\left(\beta\sigma_{j}^{2}+\rho_{\text{o}}\right)}{\left(\sigma_{j}^{2}+\rho_{\text{o}}\right)^{2}} 
+
\frac{n_{2}}{\rho_{\text{o}}}\sum_{i=1}^{n}\frac{\sigma_{i}^{2}b_{i}^{2}}{\left(\sigma_{i}^{2}+\rho_{\text{o}}\right)^{2}}.
\end{align}
Given how we choose $\epsilon$, Eq.~(\ref{eqC:MSE sum 1}) can be approximated as  
\begin{align}
\label{eqC:MSE sum 2}
G\left(\epsilon\right)
&\approx
\beta \sum_{i=1}^{n}\sigma_{i}^{-2} b_{i}^{2} \sum_{j=1}^{n_{1}}  \sigma_{j}^{-2}
-
\beta n_{1} \sum_{i=1}^{n}\sigma^{-4}_{i} b_{i}^{2} +\frac{n_{2}}{\epsilon} \sum_{i=1}^{n}\sigma_{i}^{-2} b_{i}^{2}.
\end{align}
Solving $G \left(\epsilon\right) = 0 $ from (\ref{eqC:MSE sum 2}) leads to the following root
\begin{equation}
\label{eqC:MSE sum 3}
\epsilon =  \frac{•n_{2} \sum_{i=1}^{n}\sigma_{i}^{-2} b_{i}^{2}}{\beta n_{1} \sum_{i=1}^{n}\sigma^{-4}_{i} b_{i}^{2} - \beta \sum_{i=1}^{n}\sigma_{i}^{-2} b_{i}^{2} \sum_{j=1}^{n_{1}}  \sigma_{j}^{-2}}.
\end{equation}
Now, we would like to know if the root define by (\ref{eqC:MSE sum 3}) is positive. For (\ref{eqC:MSE sum 3}) to be positive, the following relation should hold
\begin{equation}
\label{eqC:MSE sum 4}
n_{1} \sum_{i=1}^{n}\sigma_{i}^{-4} b_{i}^{2} \geq  \sum_{i=1}^{n}\sigma_{i}^{-2} b_{i}^{2} \sum_{j=1}^{n_{1}}  \sigma_{j}^{-2}.
\end{equation}  
Starting from the right hand side of (\ref{eqC:MSE sum 4}), and given that $\sigma_{1} \ge \sigma_{2} \geq \dotsi \geq \sigma_{n}$, we can bound this term as
\begin{equation}
\label{eqC:MSE sum 5}
\sum_{i=1}^{n}\sigma_{i}^{-2} b_{i}^{2} \sum_{j=1}^{n_{1}}  \sigma_{j}^{-2} \leq  \sigma_{n_{1}}^{-2} \sum_{i=1}^{n}\sigma_{i}^{-2} b_{i}^{2} \sum_{j=1}^{n_{1}}  1 = n_{1} \sigma_{n_{1}}^{-2} \sum_{i=1}^{n}\sigma_{i}^{-2} b_{i}^{2}. 
\end{equation}
On the other hand, given how we choose $n_{1}$ and $n_{2}$, we have
\begin{equation}
\label{eqC:MSE sum 6}
\sum_{i={n_{1}+1}}^{n}\sigma_{i}^{-2} \geq \sum_{i=1}^{n_{1}}\sigma_{i}^{-2},
\end{equation}
which can help us to bound the left hand side of (\ref{eqC:MSE sum 4}) as
\begin{equation}
\label{eqC:MSE sum 7}
n_{1} \sum_{i=1}^{n}\sigma_{i}^{-4} b_{i}^{2} \geq  n_{1}  \sigma_{n_{1}}^{-2} \sum_{i=1}^{n}\sigma_{i}^{-2} b_{i}^{2}.
\end{equation}
Now, from (\ref{eqC:MSE sum 5}) and (\ref{eqC:MSE sum 7}) we find that a lower bound for the left hand side of (\ref{eqC:MSE sum 4}) is equal to the upper bound for its right hand side. Then we can conclude from these two relations that
\begin{equation}
n_{1} \sum_{i=1}^{n}\sigma_{i}^{-4} b_{i}^{2} \geq  \sum_{i=1}^{n}\sigma_{i}^{-2} b_{i}^{2} \sum_{j=1}^{n_{1}}  \sigma_{j}^{-2}.
\end{equation}
Thus, $\epsilon$ is a positive root for the COPRA function in (\ref{eq:IBPRfunction}).
 
Now, we would like to know if $\epsilon$ can be considered as a value for our regularization parameter $\rho_{\text{o}}$. A direct way to prove that can be noted from the fact that having $\epsilon \ll \sigma_{i}^{2}$ \  $\forall i \in [1,n]$ will not provide any source of regularization to the problem. Hence, the RLS in (\ref{eq:R-LS}) converges to the LS in (\ref{eq:pure LS solution}).

As a remark, we can assume that the approximation in (\ref{eqC:MSE sum 2}) is uniform such that it does not affect the position of the roots. Thus, we can claim that this root is not coming from the negative region of the axis. In fact, we can easily prove that (\ref{eqC:MSE sum 1}) does not have a negative root that is close to zero. Thus, this root is not coming from the negative region as a result of the function approximation (i.e., perturbed root).

\bibliographystyle{IEEEbib}
\bibliography{References_j14}

\end{document}

%% file: macros.tex
\long\def\comment#1{}

\DeclareMathOperator*{\argmin}{arg\,min}

\newfont{\bbb}{msbm10 scaled 700}

\newfont{\bb}{msbm10 scaled 1100}

\newcommand{\bv}{{\bf b}}

\newcommand{\uv}{{\bf u}}

\newcommand{\vv}{{\bf v}}
\newcommand{\xv}{{\bf x}}
\newcommand{\yv}{{\bf y}}
\newcommand{\zv}{{\bf z}}

\newcommand{\Am}{{\bf A}}
\newcommand{\Bm}{{\bf B}}

\newcommand{\Fm}{{\bf F}}

\newcommand{\Hm}{{\bf H}}
\newcommand{\Id}{{\bf I}}

\newcommand{\Mm}{{\bf M}}

\newcommand{\Rm}{{\bf R}}

\newcommand{\Um}{{\bf U}}

\newcommand{\Vm}{{\bf V}}
\newcommand{\Xm}{{\bf X}}

\newcommand{\Sigmam}{\hbox{\boldmath$\Sigma$}}